\documentclass[a4paper]{amsart}
\usepackage[a4paper, total={7in, 9in}]{geometry}
\usepackage[english]{babel}

\usepackage{amssymb,amsmath}
\usepackage[mathscr]{euscript}
\usepackage{cleveref}
\usepackage{amsthm}
\usepackage{tikz-cd}
\usepackage{comment}
\usepackage{xcolor}
\usepackage[T1]{fontenc}

\newcommand{\CSP}[0]{\ensuremath{\textrm{CSP}}}
\newcommand{\MP}[0]{\ensuremath{\textrm{MP}}}
\newcommand{\Obs}[0]{\ensuremath{\textrm{Obs}}}
\newcommand{\Cat}[0]{\ensuremath{\textrm{Cat}}}
\newcommand{\NP}[0]{\ensuremath{\textrm{NP}}}
\newcommand{\polynomial}[0]{\ensuremath{\textrm{P}}}
\newcommand{\NPc}[0]{\ensuremath{\textrm{NP-complete}}}

\newcommand{\symb}[1]{\ensuremath{{#1}}}
\newcommand{\structure}[1]{\ensuremath{\mathbf{#1}}}
\newcommand{\matrice}[1]{\ensuremath{\mathsf{#1}}}

\newcommand{\tuple}[1]{\ensuremath{\mathbf{#1}}}
\newcommand{\rel}[1]{\ensuremath{#1}}
\newcommand{\family}[1]{\ensuremath{\mathcal{#1}}}
\newcommand{\fo}{\ensuremath{\mathsf{FO}}}

\newcommand{\true}{\ensuremath{\mathsf{true}}}

\newcommand{\VCSP}[0]{\ensuremath{\textrm{VCSP}}}

\newtheorem{theorem}{Theorem}[section]
\newtheorem{proposition}[theorem]{Proposition}
\newtheorem{lemma}[theorem]{Lemma}
\newtheorem*{claim}{Claim}

\newtheorem{corollary}[theorem]{Corollary}

\theoremstyle{definition}
\newtheorem{definition}[theorem]{Definition}

\newtheorem{example}{Example}

\theoremstyle{remark}
\newtheorem*{remark}{Remark}
\newtheorem{observation}[theorem]{Observation}

\Crefname{lemma}{Lemma}{Lemmas}
\Crefname{proposition}{Proposition}{Propositions}
\crefname{observation}{observation}{observations}

\begin{document}



\title[Generalisations of Matrix Partitions]{Generalisations of Matrix Partitions : Complexity and Obstructions}

\author[A. Barsukov]{Alexey Barsukov}
\address{Faculty of Mathematics and Physics, Charles University, 
Sokolovsk\'a 49/83, Prague, 186 75, Czechia}
\email{alexey.barsukov@matfyz.cuni.cz}

\author[M. M. Kant\'e]{Mamadou Moustapha Kant\'e}
\address{Universit\'e Clermont Auvergne, Clermont Auvergne INP, LIMOS, CNRS, 
1 Rue de la Chebarde, Aubi\`ere, 63178, France}
\email{mamadou.kante@uca.fr}
\urladdr{https://perso.isima.fr/~makante/}
\thanks{This work was mainly done when A. Barsukov was PhD student at Universit\'e Clermont Auvergne, Clermont Auvergne INP, LIMOS, CNRS, F-63000
  Clermont-Ferrand France, and was partially supported by the grant from the French National Research Agency under PRC program (DIFFERENCE: ANR-20-CE48-0002). A. Barsukov is supported by the European Union (ERC,
  POCOCOP, 101071674). Views and opinions expressed are however those of the author(s) only and do not necessarily reflect those of the European Union or the
  European Research Council Executive Agency. Neither the European Union nor the granting authority can be held responsible for them.
M.M. Kant\'e is supported by the grant from the French National Research Agency under JCJC program (ASSK: ANR-18-CE40-0025-01).}


\begin{abstract}
A \emph{trigraph} is a graph where each pair of vertices is labelled either $0$ (a non-arc), $1$ (an arc) or $\star$ (both an
arc and a non-arc). In a series of papers, Hell and co-authors (see for instance [Pavol Hell: Graph partitions with prescribed
patterns. Eur. J. Comb. 35: 335-353 (2014)]) proposed to study the complexity of homomorphisms from graphs to trigraphs, called \emph{Matrix Partition
  Problems}, where arcs and non-arcs can be both mapped to $\star$-arcs, while a non-arc cannot be mapped to
an arc, and vice-versa. Even though Matrix Partition Problems are generalisations of \emph{Constraint Satisfaction Problems
  ($\CSP$s)}, they share with them the property of being ``intrinsically'' combinatorial. So, the question of a possible
$\polynomial$-time vs $\NPc$ dichotomy is a very natural one and was raised in Hell et al.'s papers. We propose a
  generalisation of Matrix Partition Problems to relational structures and study them with respect to the question of a dichotomy. We first show that trigraph homomorphisms and Matrix Partition Problems are $\polynomial$-time equivalent, and then prove that one can also restrict
(with respect to having a dichotomy) to relational structures with a single relation. Failing in proving that Matrix Partition Problems on directed graphs are not $\polynomial$-time
equivalent to Matrix Partitions on relational structures, we give some evidence that it might be unlikely by formalising the
  reductions used in the case of $\CSP$s and by showing that such reductions cannot work for the case of Matrix Partition Problems. We
  then turn our attention to Matrix Partition problems that can be described by finite sets of (induced-subgraph) obstructions. We show, in particular, that any such problem has finitely many minimal obstructions if and only if it has finite duality.  We conclude by showing that on trees (seen as trigraphs) it is $\NPc$ to decide whether a given tree has a homomorphism to another input trigraph. The latter shows a notable difference on tractability between $\CSP$ and Matrix Partition Problems as it is well-known that $\CSP$ is tractable on the class of trees.
\end{abstract}
\maketitle








\section{Introduction}

Ladner showed in~\cite{ladner1975} that, under the assumption $\polynomial \not= \NP$, there exist problems that are neither in $\polynomial$ nor  $\NPc$. This raised the following question: which subclasses of $\NP$ admit a  \emph{P vs NP-complete dichotomy}\footnote{In the rest of the article, we omit ``$\polynomial$ vs $\NPc$'' because we study only this particular dichotomy.}, i.e., every problem of the class is either in $\polynomial$  or is $\NPc$.  For instance, Schaefer proved in his seminal paper~\cite{schaefer1978} that every Boolean $\CSP$ admits a dichotomy. Hell and Ne{\v s}et{\v r}il~\cite{hellnesetril1990} showed a similar dichotomy for homomorphism problems on  undirected graphs. The class of \emph{Constraint Satisfaction Problems} ($\CSP$s for short) is usually described (see~\cite{hellnesetril}) as the family of decision problems $\CSP(\structure{A})$, for every finite relational structure $\structure{A}$, that checks the existence of a homomorphism to $\structure{A}$ from a given input structure. Since general $\CSP$s are generalisations of both Boolean $\CSP$s and homomorphism problems on undirected graphs, researchers wondered whether a dichotomy can hold for them as well. This question, also known as the \emph{$\CSP$ conjecture}, was stated by Feder and Vardi in~\cite{federvardi1998}.
For around two decades, the $\CSP$ conjecture was verified for many special cases, see for instance~\cite{federvardi1998,bulatov2006,grohe2007}, but, more importantly, its study brings many mathematical tools in studying algorithmic and complexity questions, in particular, the algebraic tools~\cite{bulatov2005}.    Recently, Bulatov~\cite{bulatov2017} and Zhuk~\cite{zhuk2020} independently answered in the affirmative the $\CSP$ conjecture.

Motivated by the $\CSP$ conjecture, many homomorphism type problems  have been introduced and studied under the realm of a dichotomy, e.g., \emph{full homomorphism}~\cite{ballpultrnesetril2010}, \emph{locally injective/surjective homomorphism}~\cite{macgillivray2010, bodirsky_surjective_2012}, \emph{list homomorphism}~\cite{hell_list_2011}, \emph{quantified $\CSP$}~\cite{zhukmartin2020,zhukmartin2022}, \emph{infinite $\CSP$}~\cite{bodirsky2018,bodirsky_madelaine_mottet_article}, \emph{$\VCSP$}~\cite{kolmogorov2013}, etc. In this paper, we are interested in the \emph{Matrix Partition} Problem introduced in~\cite{federhellkleinmotwani2003} which finds its origin in combinatorics as other variants of the $\CSP$ conjecture, e.g., list or surjective homomorphism. 
A \emph{trigraph} is a pair $\structure{G}=(G,\rel{E}^\structure{G})$ where $\rel{E}^\structure{G}\colon G^2\to\{0,1,\star\}$. A \emph{homomorphism} between two
trigraphs $\structure{G}$ and $\structure{H}$ is a mapping $h\colon G\to H$ such that, for all $(x,y)\in G^2$,
$\rel{E}^\structure{H}(h(x),h(y))\in\{\rel{E}^\structure{G}(x,y),\star\}$, see~\cite{hellnesetriltrigraph2007}. As any graph is a trigraph, Hell et
al. (\cite{federhellkleinmotwani2003,federhell2007,hell2014}) proposed a way to consider combinatorial problems on graphs as trigraph homomorphism problems, and
called them \emph{Matrix Partition Problems}\footnote{The term \emph{Matrix Partition Problem} is a natural one because any trigraph can be represented by a
  matrix where each entry is in $\{0,1,\star\}$, and a trigraph homomorphism is a partition problem where the arcs between two parts $V_i$ and
  $V_j$ are controlled by the entry of the matrix on $(i,j)$.}. Particularly, any $\CSP$ problem on (directed) graphs can be represented as a Matrix Partition
Problem, thus the latter is a generalisation of the class $\CSP$. Motivated by the $\CSP$ conjecture, and the similarity of Matrix Partition Problems with
$\CSP$s, Hell et al.~\cite{federhell2007,hell2014} ask whether Matrix Partition Problems may satisfy a similar dichotomy as $\CSP$s.  

There are several ways to generalise Matrix Partition Problems. 
Such problems were originally defined for (directed) graphs only.
Therefore, similarly to $\CSP$, one can study Matrix Partitions over arbitrary finite relational signatures.
Another way comes from the mismatch of the types of the input (graph), and of the target (trigraph) in Matrix Partition Problems as originally defined in~\cite{federhellkleinmotwani2003}, while in $\CSP$ these types are the same.
Therefore, attempting to make these types similar, as it is for $\CSP$, we propose to consider trigraphs in the input as well.
As every graph is also a trigraph, these new problems can be seen as the old ones, where the input is extended by adding all finite trigraphs.
For every trigraph $\structure{H}$, let us denote by $\MP(\structure{H})$ the Matrix Partition Problem as defined originally in~\cite{federhellkleinmotwani2003}, i.e., inputs are graphs, and let us denote by $\MP_\star(\structure{H})$ the new Matrix Partition problem defined here, where inputs have the same type as $\structure{H}$, i.e., are trigraphs.
The class of problems of the form $\MP(\structure{H})$ (resp. $\MP_\star(\structure{H})$), for every trigraph $\structure{H}$, is denoted by $\MP$ (resp. $\MP_\star$).
Hell and Ne{\v s}et{\v r}il proved in~\cite{hellnesetriltrigraph2007} that the problems $\MP(\structure{H})$ and $\MP_\star(\structure{H})$ are $\polynomial$-time equivalent for every trigraph $\structure{H}$.
In particular, this implies that the classes $\MP$ and $\MP_\star$ agree on having the dichotomy.
While their proof uses probabilistic reductions, we provide a deterministic proof of this statement in Section~\ref{sec:mpstar=mp}.
For every input trigraph, we replace each element $x$ with a large enough set $V_x$ of new elements; for a 0-labelled or 1-labelled arc $(x,y)$ of the input, we assign the same label to every new arc from $V_x\times V_y$; and, for a $\star$-labelled arc $(x,y)$ of the input, the new arcs from $V_x\times V_y$ are labelled according to the distribution of $1$s and $-1$s in some large enough \emph{Hadamard matrix}~\cite{frankl1988}. 
\emph{Hadamard matrices} are over $\{1,-1\}$ and they have the property that any sufficiently large submatrix is  \emph{not monochromatic}, i.e., it must contain at least one $1$ and at least one $-1$. 
This property is related to a well-known \emph{expansion} property~\cite{hoory2006}. 
It has to be said that instead of Hadamard matrices one can use any graph with good expansion property and also achieve the same result for $\MP$ and $\MP_\star$.

 Feder and Vardi in~\cite{federvardi1998} showed that every $\CSP$ over a finite signature is $\polynomial$-time equivalent to a $\CSP$ on directed graphs. Bulin et al. in~\cite{bulin2015} gave a more detailed proof of this fact and showed that all the reductions are log-space. In \Cref{sec:arity}, we ask whether similar reductions exist for Matrix Partition Problems. Using the equivalence between $\MP$ and $\MP_\star$ achieved in \Cref{sec:mpstar=mp}, we show that every problem in $\MP$ over any finite signature is $\polynomial$-time equivalent to a problem in $\MP$ over relational signatures with a single relation symbol. 
 
 We then turn our attention to the $\polynomial$-time equivalence between $\MP$ on relational structures with a single relation to $\MP$ on directed graphs. While we think that, contrary to the $\CSP$ case, $\MP$ on relational structures is richer than $\MP$ on directed graphs, we fail to prove it. Instead, we analyse the type of reductions used in the $\CSP$ case~\cite{federvardi1998,bulin2015} and show that it is unlikely that such reductions work for $\MP$, unless $\MP$ is contained in $\CSP$ modulo $\polynomial$-time equivalence. In order to show this, we introduce another generalisation of Matrix Partition Problems, denoted by $\MP_\varnothing$. We first reduce every problem in $\MP$ to a $\CSP$ problem by identifying for each tuple whether it is labelled $1$ or $0$ (we introduce for each relation $\rel{R}$ two relations $\rel{R}_0$, for $0$-labelled tuples, and $\rel{R}_1$ for $1$-labelled tuples). Therefore, every $\MP$ problem is a $\CSP$ problem, but restricted to  ``complete structures'', i.e., any tuple should be either in  $\rel{R}_0$ or in $\rel{R}_1$. When we relax this completeness property, we obtain the class of problems $\MP_\varnothing$, where we introduce a new value for tuples, namely $\varnothing$, which can be mapped to any value among $\{0,1,\star\}$. Firstly, we show in \Cref{sec:fromemptysettocsp} that, for every problem in $\MP_\varnothing$, there is a $\polynomial$-time equivalent problem in $\CSP$. We later use this result to show in \Cref{sec:arity} that any reduction similar to the one from~\cite{federvardi1998,bulin2015} cannot prove the $\polynomial$-time equivalence between $\MP$ over any finite signature and $\MP$ on directed graphs, unless every $\MP$ problem is  $\polynomial$-time equivalent to a problem in $\CSP$.

 A natural way to prove that a problem is in $\polynomial$ is to show that it is described by a finite set of obstructions. In the case of $\CSP$, $\family{F}$
 is called a \emph{duality set} for $\CSP(\structure{H})$ if, for every structure $\structure{G}$, there is no homomorphism from $\structure{G}$ to $\structure{H}$ if and only if there is $\structure{F}\in \family{F}$ such that $\structure{F}$ homomorphically maps to $\structure{G}$. It is known that
 $\CSP(\structure{H})$ has a finite duality set if and only if it is definable by a \emph{first-order formula}~\cite{atserias2008}. Feder, Hell, and
 Xie proposed in~\cite{federhell2007} to study Matrix Partition Problems with finite sets of (induced-subgraph) obstructions, i.e., a
 graph admits a partition if and only if it does not have an induced subgraph that belongs to a finite family of forbidden graphs. They proposed a necessary
 (but not sufficient) condition for a matrix $\matrice{M}$ to have finitely many obstructions. Later, Feder, Hell, and Shklarsky showed in
~\cite{federhellsplit2014} that any Matrix Partition Problem has finitely many obstructions if the input consists only of split graphs. In
 \Cref{sec:obstructions}, we show that $\MP(\structure{H})$ has finitely many obstructions if and only if so does
   $\MP_\star(\structure{H})$. We also define duality sets for Matrix Partition Problems, and show that, for $\MP$ and
   $\MP_\star$ problems, the property of having a finite duality set and the property of having a finite set of obstructions are two equivalent notions.

Apart from it, we study how the finiteness of obstruction sets for $\CSP$s is related to the finiteness for trigraphs. We demonstrate that if $\MP_\varnothing(\structure{H})$ (that is, in fact, a $\CSP$, see \Cref{sec:fromemptysettocsp}) has a finite set of obstructions, then $\MP(\structure{H})$  also has a finite set of obstructions. We show that the other direction is false by giving an example of a $\star$-graph $\structure{H}$ such that $\MP_\star(\structure{H})$ has finitely many obstructions and $\MP_\varnothing(\structure{H})$ has an infinite set of obstructions.

We finally consider tractability questions.
There is another type of $\CSP$ problems, where each problem is described by a class $\family{C}$ of relational structures such that the input consists of some $\structure{A}\in\family{C}$ and an arbitrary finite structure $\structure{B}$, and the goal is to decide whether there is a homomorphism from $\structure{A}$ to $\structure{B}$. Grohe showed in~\cite{grohe2007} for relational signatures of bounded arity, that any such problem is solvable in $\polynomial$-time
if and only if all the structures from $\family{C}$ have bounded tree-width. We show that the analogous problem is $\NPc$ for the case of Matrix
Partition Problems, even when $\family{C}$ consists only of trees, by reducing $3$-SAT to it.

\medskip

{\bf Outline.} Necessary definitions are given in \Cref{sec:preliminaries}, and the $\polynomial$-time equivalence between $\CSP$ and the class
  $\MP_\varnothing$ is shown in \Cref{sec:fromemptysettocsp}. The $\polynomial$-time equivalence of the class $\MP$ and
  the class $\MP_\star$ is shown in \Cref{sec:mpstar=mp}. We prove in \Cref{sec:arity} the $\polynomial$-time equivalence
      between $\MP$ over one-relational signatures and $\MP$ over any signature. We also provide evidence against the $\polynomial$-time equivalence between
    the class $\MP$ on directed graphs and the class $\MP$ over any signature. \Cref{sec:obstructions} covers the finiteness for the
    obstruction sets. We discuss with some remarks in \Cref{sec:tractability}, the potential utility of tree-width for the $\MP$ problems.

\section{Preliminaries}\label{sec:preliminaries}

We denote by $\mathbb{N}$ the set of nonnegative integers and, for  $n\in \mathbb{N}$, we let $[n]$ be $\{1,\ldots, n\}$. Let $V$ be a set. The cardinality of $V$ is denoted by $|V|$ and its power set is denoted by $2^V$. For a positive integer $k$, elements (tuples) of $V^k$ are often represented by boldface lower case letters (e.g., $\tuple{t}$), and the $i$-th coordinate of a tuple $\tuple{t}$ is denoted by $t_i$. If $f:V\to A$ is a mapping from $V$ to a set $A$, we denote by $f(\tuple{t})$ the tuple $(f(t_1),\ldots,f(t_k))$, and by $f(X)$, for $X\subseteq V$, the set $\{f(x)\mid x\in X\}$.


Our graph terminology is standard, see for instance~\cite{Diestel12}. In this paper, we deal mostly with labelled complete relational structures, i.e., each relation of arity $k$ is $V^k$, and tuples are labelled by the elements of a partially ordered set~\cite{stanley_enumerative_combinatorics}.

\begin{definition}[$(\symb{*},\sigma)$-structures]
A \emph{signature} $\sigma$ is a set $\{\rel{R}_1,\ldots,\rel{R}_n\}$, each $\rel{R}_i$ has arity $k_i\in \mathbb{N},\ i\in [n]$.

Let $(P_*, \preceq_*)$ be a partially ordered set (poset). A \emph{$(\symb{*},\sigma)$-structure} is a tuple $\structure{G} := (G; \rel{R}_1^\structure{G},\ldots,\rel{R}_n^\structure{G})$, where $G$ is a finite set and, for each $i\in [n]$,
$\rel{R}_i^\structure{G}\colon G^{k_i} \to P_*$
is interpreted as a mapping to the elements of the poset $(P_*, \preceq_*)$.
\end{definition}

We will always denote a $(\symb{*},\sigma)$-structure by a boldface capital letter, e.g. $\structure{A}$, and its domain by the same letter in plain font, e.g. $A$. It is worth mentioning that the notion of $(\symb{*},\sigma)$-structure is different from the one in \emph{universal algebra}, where in the latter case the functional symbol $\rel{R}_i$ is interpreted in $\structure{G}$ as a function from $G^{k_i}\to G$. 

For a $(\symb{*},\sigma)$-structure $\structure{G}$ and $X\subseteq G$, the substructure of  $\structure{G}$ induced by $X$ is the $(\symb{*},\sigma)$-structure $\structure{G'}$ with domain $G'=X$ and, for $\rel{R}\in \sigma$ of arity $k$ and $\tuple{t}\in X^k$, $\rel{R}^{\structure{G'}}(\tuple{t})=\rel{R}^{\structure{G}}(\tuple{t})$. We denote by $\structure{G}\setminus X$ the substructure of $\structure{G}$ induced by $G\setminus X$.

We now extend the notion of homomorphism between relational structures to $(\symb{*},\sigma)$-structures, the difference being the ability to map a tuple to a ``greater'' one.

\begin{definition}[homomorphism for $(\symb{*},\sigma)$-structures] For two $(\symb{*},\sigma)$-structures $\structure{G}$ and $\structure{H}$, a  mapping $h\colon G\to H$ is called a \emph{homomorphism from $\structure{G}$ to $\structure{H}$} if, for each $\rel{R} \in \sigma$ of arity $k$, and $\tuple{t}\in G^{k}$, we have that
$\rel{R}^\structure{G}(\tuple{t}) \preceq_\symb{*} \rel{R}^\structure{H}(h(\tuple{t}))$.

As usual, we will write  $h\colon \structure{G} \to \structure{H}$ to mean that $h\colon G\to H$ is a homomorphism from $\structure{G}$ to $\structure{H}$. We say that $h:\structure{G}\to \structure{H}$ is  \emph{surjective} (resp. \emph{injective}) if $h:G\to H$ is surjective (resp. injective).
\end{definition}

We can now explain how the notion of homomorphism between  $(\symb{*},\sigma)$-structures subsumes the usual ones. First, let us introduce the  partial orders that we consider in this paper. See Figure~\ref{fig:posets} for their Hasse diagrams.
\begin{itemize}
    \item $(P_{01}, \preceq_{01})$, where $P_{01}=\{0,1\}$ and $\preceq_{01}$ is the empty order with $0$ and $1$ incomparable.
    
    \item $(P_\CSP, \preceq_\CSP)$, where $P_\CSP=\{0,1\}$ and $\preceq_\CSP$ is a total order with $0\preceq_\CSP 1$. 
    
    \item $(P_\star, \preceq_\star)$, where $P_\star=\{0,1,\star\}$ and $\preceq_{\star}$ is the poset with $0\preceq_{\star} \star$ and $1\preceq_{\star} \star$, and $0$ is incomparable with $1$.
    
    \item $(P_\varnothing, \preceq_\varnothing)$ where $P_\varnothing=\{\varnothing,0,1,\star\}$ and $\preceq_{\varnothing}$ is the poset with $\varnothing\preceq_\varnothing 0\preceq_\varnothing \star$ and $\varnothing \preceq_\varnothing 1 \preceq_\varnothing \star$, and $0$ is incomparable with $1$.
\end{itemize}

\begin{figure}
\centering
\includegraphics[scale=1]{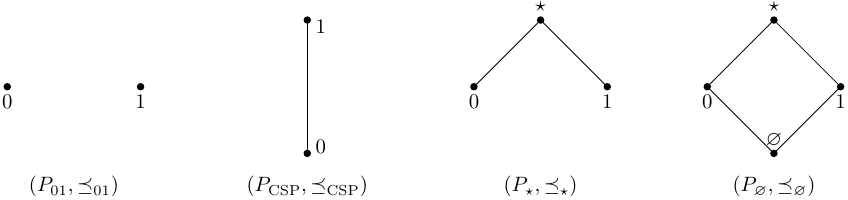}
\caption{Hasse diagrams of the four posets.}
\label{fig:posets}
\end{figure}

\begin{remark}
If the signature $\sigma$ is clear from the context,  then we will just write \emph{$*$-structure} instead of $(\symb{*},\sigma)$-structure, for $*\in \{01,\star,\varnothing\}$. Also, if $\sigma = \{E(\cdot,\cdot)\}$, then we will write \emph{$*$-graph} instead.  Finally, we will talk about \emph{relational $\sigma$-structures} and \emph{directed graphs}, instead of $(\symb{\CSP},\sigma)$-structures and $\symb{\CSP}$-graphs. Furthermore, for any tuple (arc) $\tuple{t}\in A^k$ of a $*$-structure ($*$-graph) $\structure{A}$ corresponding to a symbol $R\in\sigma$ that is clear from the context, we will call $\tuple{t}$ a \emph{$v$-tuple ($v$-arc)} if $\rel{R}^\structure{A}(\tuple{t})=v$ for some element $v$ of the poset $(P_*, \preceq_*)$.
\end{remark}

It is not hard to check that the notion of $(\symb{\CSP},\sigma)$-structures is equivalent to the usual notion of relational $\sigma$-structures, and
homomorphisms between $(\symb{\CSP},\sigma)$-structures are equivalent to usual homomorphisms. Notice that homomorphisms between
$(\symb{01},\sigma)$-structures are exactly full homomorphisms on relational structures. The following immediately follows from the definitions.


\begin{proposition}\label{prop:subposet}
Let $(P_*, \preceq_*)$ and $(P_{*'}, \preceq_{*'})$ be two posets, with $(P_*, \preceq_*)$ a subposet of $(P_{*'}, \preceq_{*'})$. Then, every $(\symb{*},\sigma)$-structure is also a $(\symb{*'},\sigma)$-structure, for any $\sigma$.
\end{proposition}

Particularly, for any $\sigma$, every $(\symb{{01}},\sigma)$-structure is a $(\symb{\star},\sigma)$-structure, and every $(\symb{\star},\sigma)$-structure is a $(\symb{\varnothing},\sigma)$-structure. For $*\in\{01,\star,\varnothing\}$, we denote by $\Cat_*$ the set of
all $(\symb{*},\sigma)$-structures\footnote{We use the notation $\Cat_*$ because one can use $(\symb{*},\sigma)$-structures as objects and 
  homomorphisms as arrows to make a category.}. From Proposition~\ref{prop:subposet} and the definitions of $(P_{01},\preceq_{01}),(P_\star,\preceq_\star)$, and $(P_\varnothing,\preceq_\varnothing)$, we have the following inclusion: $\Cat_{01}^\sigma\subset \Cat_\star^\sigma \subset \Cat_\varnothing^\sigma.$

We can now define the homomorphism problems, that we restrict for conciseness to the four posets from Figure~\ref{fig:posets}.

\begin{definition}[Generalised Matrix Partition]
Let $\sigma$ be a finite signature and $*$ be in $\{01,\star,\varnothing\}$. For a $(\star,\sigma)$-structure $\structure{H}$, the problem $\MP^\sigma_*(\structure{H})$
denotes the set of all $*$-structures $\structure{G}$ such that there exists a homomorphism $h\colon\structure{G}\to\structure{H}$. We always omit subscript $01$ in $\MP_{01}^\sigma(\structure{H})$ and write $\MP^\sigma(\structure{H})$ instead.
\end{definition}

For a signature $\sigma$, $\MP^\sigma$, $\MP_\star^\sigma$, and $\MP_\varnothing^\sigma$ denote the classes of problems $\MP^\sigma(\structure{H})$,
$\MP_\star^\sigma(\structure{H})$, and $\MP_\varnothing^\sigma(\structure{H})$ respectively, for all $(\star,\sigma)$-structures $\structure{H}$. If $\sigma = \{\rel{E}(\cdot,\cdot)\}$ -- the directed graph signature, then we will omit the $\sigma$-superscript and will just write $\MP$, $\MP_\star$, and $\MP_\varnothing$.

If $\structure{H}$ is
a relational $\sigma$-structure, then we write  $\CSP^\sigma(\structure{H})$ for the set of all relational $\sigma$-structures $\structure{G}$ such that there exists a homomorphism
$h\colon\structure{G}\to\structure{H}$. 


We now give the original definition of Matrix Partition Problems given by Feder et al. in~\cite{federhell2007}. Let $\matrice{M}$ be an $n\times n$-matrix with entries from $\{0,1,\star\}$.  A graph $\structure{G}$ admits an \emph{$\matrice{M}$-partition} if there is a function $m:G\to [n]$ such that, for all \emph{distinct} $x,y\in G$, $\rel{E}^{\structure{G}}(x,y)\preceq_\star \matrice{M}[m(x),m(y)]$.

\begin{remark}
The definition from~\cite{federhell2007} and our definition of $\MP(\structure{H})$ are not the same. Unlike Feder et al., we consider all possible graphs in the input, not only the loopless ones. This implies that we do not need to require that $x,y\in G$ must be distinct to satisfy the condition of Matrix Partition. We decided to use our definition because it can be generalised better. One of the reasons is the ambiguity of what it means to be ``distinct'' when the arity is greater than 2: it may be ``pairwise distinct'' or ``not all equal''. Another reason is that, in our definition, homomorphisms are transitive, so one can consider $\star$-structures as objects of the category $\Cat_\star$, where arrows are homomorphisms, similarly to the category of relational $\sigma$-structures associated with $\CSP^\sigma$ problems.
\end{remark}

\begin{example}\label{ex:one}
A graph is called \emph{split} if there exists a partition of its vertices into two classes such that one class induces an independent set and the other class induces a clique. The problem that decides whether a given input graph is split is a standard example of a Matrix Partition: it is split if and only if it admits an $\matrice{M}$-partition, where
\begin{equation*}
\matrice{M} = \begin{pmatrix}0 & \star\\ \star & 1\end{pmatrix}
\end{equation*}
This problem is solvable in $\polynomial$-time because it can be reduced to 2-SAT. Let $\structure{G}$ be an input graph with vertices $G=\{g_1,\ldots,g_n\}$. Let $X:=\{x_1,\ldots,x_n\}$ be a set of variables. Let $\Phi_\structure{G}$ be the 2-SAT formula with variables from $X$ such that, for every two distinct $g_i,g_j\in G$, if $g_ig_j$ is an edge of $\structure{G}$, then $\Phi_\structure{G}$ contains a clause $(x_i\vee x_j)$; otherwise, $\Phi_\structure{G}$ contains a clause $(\neg x_i\vee \neg x_j)$. It can easily be checked that $\structure{G}$ is a split graph if and only if $\Phi_\structure{G}$ is satisfiable.
\begin{figure}
\centering
\includegraphics[scale=1]{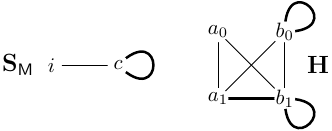}
\caption{Symmetric $\star$-graphs from Examples~\ref{ex:one} and~\ref{ex:two}. 1-edges are thick, $\star$-edges are thin.}
\label{fig:example}
\end{figure}
The corresponding $\star$-graph $\structure{S}_\matrice{M}$ is given on Figure~\ref{fig:example}. Feder and Hell's Matrix Partitions are always harder than the corresponding $\MP$ problems. Particularly, $\MP(\structure{S}_\matrice{M})$ reduces to $\matrice{M}$-partition as follows. For a given input 01-graph $\structure{G}$, replace every 0-loop vertex $x$ with two non-adjacent vertices $x_1,x_2$ that are \emph{twins}, i.e., for every $y\in G$, $\rel{E}^\structure{G}(x_1,y)=\rel{E}^\structure{G}(x_2,y)$ and $\rel{E}^\structure{G}(y,x_1)=\rel{E}^\structure{G}(y,x_2)$. 1-loop vertices are replaced by two adjacent twins. Denote the resulting graph by $\structure{G}'$. If $\structure{G}\to\structure{S}_\matrice{M}$, then the vertices of $\structure{G}'$ obtained from 0-loops induce an independent set, and those that are obtained from 1-loops induce a clique, so $\structure{G}'$ is a split graph. If $\structure{G}'$ is a split graph, then, for each pair of non-adjacent twins, at least one of them belongs to the independent set, and, as they are twins, we can put the other twin to the independent set as well. A similar argument holds for adjacent twins. Then, there is a homomorphism $\structure{G}\to\structure{S}_\matrice{M}$.
\end{example}
\begin{example}\label{ex:two}
Consider the  $\star$-graph $\structure{H}$ given on Figure~\ref{fig:example}.
We reduce the problem $\MP(\structure{H})$ to 2-SAT as well. Let $\structure{G}$ be an input 01-graph. It is useful to partition its elements into two parts $A\sqcup B$ depending on the loop type, i.e., for each $x\in G$, if $\rel{E}^\structure{G}(x,x)=0$, then $x\in A$, otherwise, $x\in B$.
It can be easily checked that if $h\colon \structure{G}\to\structure{H}$ is a homomorphism, then, for every $x\in G$, we have that $x\in A$ if and only if $h(x)\in\{a_0,a_1\}$.
If, for some $x,y\in A$, $\rel{E}^\structure{G}(x,y)=1$, then $h(x)\not= h(y)$, and similarly, for $x',y'\in B$, if $\rel{E}^\structure{G}(x',y')=0$, then $h(x')\not=h(y')$.
Now, we construct the formula $\Phi_\structure{G}$ on the variables $A\sqcup B$. The clauses of $\Phi_\structure{G}$ are the following.
\begin{itemize}
\item For $x,y\in A$, if $\rel{E}^\structure{G}(x,y)=1$, then we add to $\Phi_\structure{G}$ two clauses $(x\vee y)\wedge(\neg x\vee \neg y)$, which means ``$x$ is different from $y$''.
\item Similarly, for $x,y\in B$, if $\rel{E}^\structure{G}(x,y)=0$, then we also add to $\Phi_\structure{G}$ two clauses $(x\vee y)\wedge(\neg x\vee \neg y)$.
\item For $x\in A$ and $y\in B$, if $\rel{E}^\structure{G}(x,y) = 0$, then we cannot map $x$ and $y$ to $a_1$ and $b_1$. Therefore, we add $(\neg x\vee \neg y)$.
\item For $x\in A$ and $y\in B$, if $\rel{E}^\structure{G}(x,y) = 1$, then we cannot map $x$ and $y$ to $a_0$ and $b_0$. Therefore, we add $(x\vee y)$.
\end{itemize}
It can be easily checked that there is a homomorphism $h\colon\structure{G}\to\structure{H}$ if and only if $\Phi_\structure{G}$ is satisfiable.
\end{example}

Let us end these preliminaries with the notion of $\polynomial$-time equivalence between two families of problems, which allows to transfer dichotomy
results. Two decision problems $P_1$ and $P_2$ are \emph{$\polynomial$-time equivalent} if there is a $\polynomial$-time reduction from $P_1$ to $P_2$, and a $\polynomial$-time reduction from $P_2$ to $P_1$.

For two families $\family{C}$ and $\family{C}'$ of decision problems, we say that they are \emph{$\polynomial$-time equivalent} if, for every problem
$M\in \family{C}$, one can find in $\polynomial$-time $M'\in \family{C}'$ and both are $\polynomial$-time equivalent, and similarly, for every  $M'\in
\family{C}'$, one can find in $\polynomial$-time $M \in \family{C}$ and both are $\polynomial$-time equivalent. 

\begin{remark} All along the paper, whenever we consider a problem $\MP_*^\sigma(\structure{H})$, for $*\in
  \{01,\star,\varnothing\}$, we assume that there is no  $x\in H$  such that for all $\rel{R}\in \sigma$,
  $\rel{R}(x,\ldots,x)= \star$. Otherwise, the problem is trivial as then $\MP_*^\sigma(\structure{H})$ equals  $\Cat_*^\sigma$. 
\end{remark}

\section{$\MP_\varnothing^\sigma$ is contained in  $\CSP^{\sigma_\CSP}$}
\label{sec:fromemptysettocsp}

Let $\sigma=\{\rel{R}_1,\ldots,\rel{R}_n\}$ be a signature, the arity of each $\rel{R}_i$ is denoted by $k_i$. We prove in this section that there is a signature $\sigma_\CSP$ such that
any problem in $\MP_\varnothing^\sigma$ is $\polynomial$-time equivalent to a problem in $\CSP^{\sigma_\CSP}$ which implies that $\MP_\varnothing^\sigma$ has a dichotomy. It will follow from a more general result that states that every decision problem that checks the existence of a homomorphism to a fixed $(*,\sigma)$-structure is $\polynomial$-time equivalent to some finite $\CSP$ under the assumption that the poset $(P_*,\preceq_*)$ is a lattice.

Let $J$ be the set of join-irreducible elements of $(P_*,\preceq_*)$. Then, the signature $\sigma_\CSP$ is defined as follows:
\begin{equation*}
    \sigma_\CSP := \{\rel{R}_j\mid  \rel{R}\in\sigma, j\in J\}\text{, where $\rel{R}_j$ has the same arity as $\rel{R}$.}
\end{equation*}

\paragraph*{Construction}
Let $\structure{A}$ be a $(*,\sigma)$-structure. Denote by $\CSP_*^\sigma(\structure{A})$ the problem that decides if there is a homomorphism from an input $(*,\sigma)$-structure to $\structure{A}$. The corresponding relational $\sigma_\CSP$-structure $\structure{A}_\CSP$ has the same domain $A$. For every $k$-ary relation $\rel{R}\in\sigma$, every join-irreducible $j\in J$, and every tuple $\tuple{a}\in A^k$, we define $\rel{R}_{j}^{\structure{A}_\CSP}$ as follows:
\begin{equation}\label{eq:emptysetcsp_equation}
\rel{R}_{j}^{\structure{A}_\CSP}(\tuple{a})=1 \Leftrightarrow j \preceq_* \rel{R}^{\structure{A}}(\tuple{a})
\end{equation}
\begin{observation}
    The correspondence between $(*,\sigma)$-structures and relational $\sigma_\CSP$-structures is one-to-one if and only if $(P_*,\preceq_*)$ is a Boolean lattice.
\end{observation}
\begin{theorem}\label{thm:lattice_csp}
If $(P_*,\preceq_*)$ is a lattice, then, for every $(*,\sigma)$-structure $\structure{A}$, the problems $\CSP_*^\sigma(\structure{A})$ and $\CSP^{\sigma_\CSP}(\structure{A}_\CSP)$ are $\polynomial$-time equivalent.
\end{theorem}
\begin{proof}
    Let $\structure{B}$ be an input of the problem $\CSP_*^\sigma(\structure{A})$ and let $\structure{B}_\CSP$ be the corresponding $\CSP$ structure constructed in a similar way as $\structure{A}_\CSP$.
    Let $h\colon B\to A$ be a mapping.
    Then, $h$ is a homomorphism from $\structure{B}$ to $\structure{A}$ if and only if it is a homomorphism from $\structure{B}_\CSP$ to $\structure{A}_\CSP$.
    Indeed, suppose that $h\colon \structure{B}\to\structure{A}$ is a homomorphism; then, for every $k$-ary relation $\rel{R}_j\in\sigma_\CSP$ and every $\tuple{b}\in B^k$, by \cref{eq:emptysetcsp_equation}, we have the following:
    \begin{equation*}
    \rel{R}_j^{\structure{B}_\CSP}(\tuple{b})=1\Leftrightarrow j \preceq_* \rel{R}^{\structure{B}}(\tuple{b}) \Rightarrow j\preceq_* \rel{R}^{\structure{A}}(h(\tuple{b}))\Leftrightarrow \rel{R}_j^{\structure{A}_\CSP}(h(\tuple{b}))=1
    \end{equation*}
    Similarly, one can prove the other direction. So, $\CSP_*^\sigma(\structure{A})$ reduces to $\CSP^{\sigma_\CSP}(\structure{A}_\CSP)$.

    Let $\structure{G}$ be a relational $\sigma_\CSP$-structure which is an input of $\CSP^{\sigma_\CSP}(\structure{A}_\CSP)$. Consider some $k$-ary $\rel{R}\in\sigma$ and some $\tuple{g}\in G^k$. Let $J_\tuple{g}:=\{j\in J\mid  \rel{R}_j^\structure{G}(\tuple{g}) = 1\}$. Let $p_\tuple{g}:=\bigvee_{j\in J_\tuple{g}}j$ if $J_\tuple{g}$ is not empty; otherwise, let $p_\tuple{g}$ be the minimal element of $(P_*,\preceq_*)$.

    Now we will construct a $(*,\sigma)$-structure $\structure{G}_*$ such that $\structure{G}\to\structure{A}_\CSP$ if and only if $\structure{G}_*\to\structure{A}$.
    It has the same domain $G$. For every $k$-ary $\rel{R}\in\sigma$ and every $\tuple{g}\in G^k$, we put $\rel{R}^{\structure{G}_*}(\tuple{g}):=p_\tuple{g}$.
    
    If there is a homomorphism $h\colon \structure{G}\to\structure{A}_\CSP$, then, for all $k$-ary $\rel{R}\in\sigma$, $\tuple{g}\in G^k, j\in J_\tuple{g}$, we have that  $\rel{R}_j^{\structure{A}_\CSP}(h(\tuple{g}))=1$.
    So, by \cref{eq:emptysetcsp_equation}, for all $j\in J_\tuple{g}$, $j\preceq_* \rel{R}^\structure{A}(h(\tuple{g}))$. Therefore, $p_\tuple{g}\preceq_*\rel{R}^\structure{A}(h(\tuple{g}))$. As $\rel{R}^{\structure{G}_*}(\tuple{g})=p_\tuple{g}$, we conclude that $h$ is a homomorphism from $\structure{G}_*$ to $\structure{A}$.
    
    If there is a homomorphism $h\colon\structure{G}_*\to\structure{A}$, then, for every $k$-ary $\rel{R}\in\sigma, j\in J, \tuple{g}\in G^k$,
    \begin{equation*}
        \rel{R}_j^\structure{G}(\tuple{g}) = 1\Rightarrow j\preceq_* p_\tuple{g}\Leftrightarrow j\preceq_*\rel{R}^{\structure{G}_*}(\tuple{g})\Rightarrow j \preceq_* \rel{R}^\structure{A}(h(\tuple{g}))\Leftrightarrow \rel{R}_j^{\structure{A}_\CSP}(h(\tuple{g}))=1
    \end{equation*}
\end{proof}

Theorem~\ref{thm:lattice_csp} implies, in particular, that $\MP_\varnothing^\sigma$ is equivalent to a fragment of $\CSP^{\sigma_\CSP}$. Therefore, $\MP_\varnothing^\sigma$ has a dichotomy.
\begin{corollary}\label{cor:mptocsp}
For every $(\star,\sigma)$-structure $\structure{H}$, there is a relational $\sigma_\CSP$-structure $\structure{H}_\CSP$ such that $\MP_\varnothing^\sigma(\structure{H})$ and $\CSP^{\sigma_\CSP}(\structure{H}_\CSP)$ are $\polynomial$-time equivalent.
\end{corollary}
\begin{proof}
First, observe that, by Proposition~\ref{prop:subposet}, $\structure{H}$ is also a $\varnothing$-structure. Then, as $(P_\varnothing,\preceq_\varnothing)$ is a lattice, the result follows from Theorem~\ref{thm:lattice_csp}.
\end{proof}

For $*\in\{01,\star,\varnothing\}$, the notion of homomorphism between  $(\symb{*},\sigma)$-structures admits a notion of the \emph{core}. It extends the definition given for trigraphs ($\star$-graphs) from~\cite{hellnesetriltrigraph2007}.
%
For $*\in\{01,\star,\varnothing\}$, a $(\symb{*},\sigma)$-structure $\structure{C}$ is called a \emph{core} if any homomorphism $h\colon \structure{C}\to \structure{C}$ is an isomorphism,
where isomorphism between $(\symb{*},\sigma)$-structures is a one-to-one mapping such that it is a homomorphism and its inverse is also a homomorphism. The proof of the following proposition relies on the well-known fact that every finite relational structure is homomorphically equivalent to one of its induced substructures that is a core and that is unique up to isomorphism.

\begin{proposition}
Let $*\in\{01,\star,\varnothing\}$ and $\sigma$ be a finite relational signature. For every $(\symb{*},\sigma)$-structure $\structure{A}_*$, there exists a $(\symb{*},\sigma)$-structure $\structure{C}_*$ that is a core and that is homomorphically equivalent to $\structure{A}_*$. Such a structure is unique up to isomorphism.
\end{proposition}

\begin{proof}
We know that $\structure{A}_*$ is also a $\varnothing$-structure by Proposition~\ref{prop:subposet}. Then, consider the relational $\sigma_\CSP$-structure $\structure{A}_\CSP$ provided by Corollary~\ref{cor:mptocsp}. It admits a core $\structure{C}_\CSP$ which is an induced substructure of $\structure{A}_\CSP$. Let $\structure{C}_*$ be the corresponding $\varnothing$-structure by Corollary~\ref{cor:mptocsp}, it must also be homomorphically equivalent to $\structure{A}_*$ and be an induced substructure of $\structure{A}_*$. As $\structure{C}_*$ is an induced substructure of $\structure{A}_*$, it is also a $*$-structure. Let $e:C_*\to C_*$ be a non-injective endomorphism. Then, the same map $e$ will be a non-injective endomorphism of the core $\structure{C}_\CSP$, which is impossible. So, $\structure{C}_*$ is a core. Let $\structure{C}_*'$ be another core of $\structure{A}_*$, that is not isomorphic to $\structure{C}_*$. But, then $\structure{C}_\CSP'$ must be the core of $\structure{A}_\CSP$ and $\structure{C}_\CSP\not\cong\structure{C}_\CSP'$, which is impossible as cores of relational structures are unique up to isomorphism.
\end{proof}

\section{Equivalence between $\MP^\sigma_\star$ and $\MP^\sigma$}\label{sec:mpstar=mp}

In this section, we will prove the following theorem.

\begin{theorem}\label{thm:mpstardichotomy}
For any finite signature $\sigma$, $\MP^\sigma$ and $\MP_\star^\sigma$ are $\polynomial$-time equivalent. 
\end{theorem}

In order to prove the $\polynomial$-time equivalence, we will show that, for any $\star$-structure $\structure{H}$, the two corresponding problems
$\MP^\sigma(\structure{H})$ and $\MP_\star^\sigma(\structure{H})$ are $\polynomial$-time equivalent.  Hell and Ne{\v s}et{\v r}il proved,  using probabilistic arguments in~\cite{hellnesetriltrigraph2007} that, for any $\star$-graph $\structure{G}$, there is a $01$-graph $\structure{G}_{01}$ such that $\structure{G}\in\MP_\star(\structure{H})\Leftrightarrow \structure{G}_{01}\in\MP(\structure{H})$. We provide in this section deterministic  $\polynomial$-time reductions based on \emph{Hadamard matrices}.

\begin{definition}[Hadamard Matrices]
An $n\times n$-matrix $\matrice{H}_n$, which entries are from $\{1,-1\}$, is called a \emph{Hadamard matrix} if
\begin{displaymath}
    \matrice{H}_n\cdot \matrice{H}_n^T = n\cdot \matrice{I}_n,
\end{displaymath}
where $\matrice{I}_n$ is the identity matrix of size $n$, and $\matrice{H}^T$ is the transpose of $\matrice{H}$.

\end{definition}

Hadamard matrices exist for any $n$ that is a power of $2$. 

\begin{lemma}[\cite{walsh1923}]\label{lem:hadamard} For every positive integer $n>1$, one can construct in time $2^{poly(n)}$ a $2^n\times 2^n$-Hadamard matrix.
\end{lemma}

Suppose that $\matrice{H}_n$ is an $n \times n$ Hadamard matrix and that its rows and columns are indexed by $[n]$. For two subsets $A,B$ of $[n]$, denote by $\matrice{H}_n[A,B]$ the submatrix of $\matrice{H}_n$ with rows indexed by $A$ and columns indexed by $B$. If all the entries of
$\matrice{H}_n[A,B]$ are equal, then $\matrice{H}_n[A,B]$ is called \emph{monochromatic}~\cite{alon1986,pudlak1988}. We will need the following to prove that if
$\structure{G}_{01}\in \MP(\structure{H})$, then $\structure{G}\in\MP_\star(\structure{H})$.

\begin{lemma}[\cite{alon1986,pudlak1988}]\label{lem:monochrome}
  Let $\matrice{H}_n$ be an $n\times n$-Hadamard matrix, whose rows and columns are indexed by $[n]$. Then, for any two disjoint sets
  $A,B\subseteq [n]$ such that $|A|=|B| > \sqrt{n}$, the submatrix $\matrice{H}_n[A,B]$ of $\matrice{H}_n$ is not monochromatic.
  %
\end{lemma}

We will prove a more general result that yields Theorem~\ref{thm:mpstardichotomy}.

\begin{lemma}\label{lem:mpstardichotomy}
Let $\sigma$ be a finite relational signature and let $\structure{H}$ be a fixed $(\star,\sigma)$-structure. Then, for any $(\star,\sigma)$-structure $\structure{G}$, one can construct in time polynomial in $|G|$ a $(01,\sigma)$-structure $\structure{G}_{01}$ such that
\begin{itemize}
\item there is a surjective homomorphism $\structure{G}_{01}\to\structure{G}$, and
\item if there is a homomorphism $\structure{G}_{01}\to\structure{H}$, then there is a homomorphism $\structure{G}\to\structure{H}$.
\end{itemize}
\end{lemma}

\begin{proof}

Denote by $m$ the size $|H|$ of $\structure{H}$. Let $n$ be the smallest positive integer such that $2^n > 4m^2+1$, and let $\matrice{H}_{2^n}$ be the Hadamard
matrix provided by Lemma~\ref{lem:hadamard}. Let  the domain of $\structure{G}_{01}$ be the disjoint union $\bigsqcup_{g\in G} V_g$, where for all $g\in G$,
$|V_g| = 2^n$. Let us enumerate the set $V_g$ as $\{v_{g,1},\ldots, v_{g,2^n}\}$ for each $g\in G$. For each $k$-ary $\rel{R}\in\sigma$ and for each tuple
$(v_{g_1,j_1},v_{g_2,j_2}, \ldots,v_{g_{k},j_{k}})\in (G_{01})^k$,
\begin{align*}
  \rel{R}^{\structure{G}_{01}}(v_{g_1,j_1}, \ldots,v_{g_{k},j_{k}}) &= \begin{cases} \rel{R}^{\structure{G}}(g_1,\ldots,g_{k}) & \textrm{if $\rel{R}^{\structure{G}}(g_1,\ldots,g_{k})\ne \star$},\\
    (\matrice{H}_{2^n}[j_1,j_2] +1)/2 & \textrm{otherwise}. \end{cases}
  \end{align*}
  Notice that, in the case ``$\rel{R}^\structure{G}(g_1,\ldots,g_{k}) = \star$'', if $\matrice{H}_{2^n}[j_1,j_2]=1$, then $\rel{R}^{\structure{G}_{01}}(v_{g_1,j_1}, \ldots,v_{g_{k},j_{k}}) = 1$, and if $\matrice{H}_{2^n}[j_1,j_2]=-1$, then $\rel{R}^{\structure{G}_{01}}(v_{g_1,j_1}, \ldots,v_{g_{k},j_{k}}) = 0$.

By construction, there exists a surjective homomorphism $\pi\colon \structure{G}_{01}\to \structure{G}$ such that, for all $g\in G$ and all $v\in V_g$, $\pi(v) = g$.

Suppose that there exists a homomorphism $h_{01}\colon \structure{G}_{01}\to\structure{H}$. By pigeonhole principle, every $V_g$ has a subset $A_g$ such that $|A_g|\geq\frac{|V_g|}{m}$ and all elements of $A_g$ are mapped to the same element $x_g$ of $\structure{H}$. Let $h\colon G\to H$ be defined as follows: for every $g\in G$, put  $h(g) := x_g$. 

Let $\rel{R}$ be some $k$-ary relation of $\sigma$, and let $\tuple{g}:=(g_1,\ldots,g_k)\in G^k$. If $\rel{R}^\structure{G}(\tuple{g})\in\{0,1\}$, then, for all $\tuple{g}_{01}\in A_{g_1}\times\dots\times A_{g_{k}}$, we have $\rel{R}^{\structure{G}_{01}}(\tuple{g}_{01})=\rel{R}^\structure{G}(\tuple{g})$. This implies that $\rel{R}^\structure{G}(\tuple{g})\preceq_\star\rel{R}^\structure{H}(h(\tuple{g}))$. If $\rel{R}^\structure{G}(\tuple{g})=\star$, then the matrix $\structure{H}_{2^n}[A_{g_1},A_{g_2}]$ has size at least $\frac{|V_g|}{m}\times \frac{|V_g|}{m}$, where $A_g$ is identified with the set $\{i\in [2^n]\mid v_{g,i}\in A_g\}$. One checks easily that there are subsets $B_1$ of $A_{g_1}$, and $B_2$ of $A_{g_2}$, that do not intersect and both are of size at least $\frac{|V_g|}{2m}$. Observe also that
\begin{displaymath}
    \frac{|V_g|}{2m} \geq \frac{2^n}{2m} \geq \frac{2^n}{\sqrt{2^n}} \geq \sqrt{2^n}
\end{displaymath}
because $2^n> 4m^2+1$, \emph{i.e.}, $\sqrt{2^n}>\sqrt{4m^2+1}> 2m$. Thus, by Lemma~\ref{lem:monochrome}, the submatrix $\matrice{H}_{2^n}[B_1,B_2]$ is not monochromatic. This means that $A_{g_1}\times\dots\times A_{g_{k}}$ contains two tuples $\tuple{g}_{01}$ and $\tuple{g}_{01}'$ such that $\rel{R}^{\structure{G}_{01}}(\tuple{g}_{01})=0$ and $\rel{R}^{\structure{G}_{01}}(\tuple{g}_{01}')=1$. Thus, $\rel{R}^\structure{H}(h(\tuple{g}))=\star$, and we are done.
\end{proof}

\begin{proof}[Proof of Theorem~\ref{thm:mpstardichotomy}]
By Proposition~\ref{prop:subposet}, every $01$-structure is also a $\star$-structure, therefore  $\MP^\sigma(\structure{H})$ trivially reduces to $\MP_\star^\sigma(\structure{H})$. For the opposite direction, by Lemma~\ref{lem:mpstardichotomy}, for every input $\star$-structure $\structure{G}$ of the problem $\MP_\star^\sigma(\structure{H})$, one can construct in time polynomial in $|G|$ a structure $\structure{G}_{01}$ such that $\structure{G}\to\structure{H}$ if and only if $\structure{G}_{01}\to\structure{H}$.
\end{proof}

\section{Arity Reduction}\label{sec:arity}

Recall that a  \emph{primitive-positive formula} $\varphi( x_1,\ldots,x_n)$  is a first-order formula ($\fo^\sigma$) of the form
\begin{displaymath}
\exists x_{n+1},\ldots,x_m\;(\psi_1\wedge\dots\wedge\psi_\ell)
\end{displaymath}
where each $\psi_i$ is either $x_s=x_j$, $\true$, or $\rel{R}(x_{i_1},\ldots,x_{i_k})=1$, with $\rel{R}$ a $k$-ary relation symbol in $\sigma$.

Let $\sigma = \{\rel{R}_1,\ldots,\rel{R}_n\},\sigma'=\{\rel{S}_1,\ldots,\rel{S}_m\}$ be two signatures, and
$\structure{A}, \structure{A}'$ be relational $\sigma$- and $\sigma'$-structures over the same domain $A$. We say that $\structure{A}$
\emph{pp-defines} $\structure{A}'$ if, for every $k$-ary relation $\rel{S}_j^{\structure{A'}}$ of $\structure{A}'$, there
exists a primitive-positive formula $\varphi_j\in \fo^\sigma$ with $k$ free variables such that, for all
$ (a_1,\ldots,a_k)\in A^{k}$, $\rel{S}_j^{\structure{A'}}(a_1,\ldots,a_k)=1 \Leftrightarrow \structure{A}'\models\varphi_j(
a_1/x_1,\ldots,a_k/x_k)$.

\begin{theorem}[\cite{jeavons1998}]\label{thm:pp-definition}
If a relational $\sigma$-structure $\structure{A}$ pp-defines a relational $\sigma'$-structure $\structure{A}'$, then the problem $\CSP^{\sigma'}(\structure{A}')$ reduces in $\polynomial$-time to $\CSP^\sigma(\structure{A})$.
\end{theorem}

\subsection{From directed graphs to many relations}

Let $\sigma = \{\rel{R}_1,\ldots,\rel{R}_n\}$ be a finite signature with arities $k_1,\ldots,k_n$, and such that $k_1\geq 2$.  We show that the
existence of a dichotomy for the class of problems $\MP_\star^\sigma$ implies the existence of a dichotomy for the class of $\star$-graphs $\MP_\star$. Let $\gamma = \{E(\cdot,\cdot)\}$ be the directed graph signature and let $\gamma_\CSP=\{E_0(\cdot,\cdot),E_1(\cdot,\cdot)\}$ be obtained from $\gamma$ by the construction from \Cref{sec:fromemptysettocsp}.

\begin{theorem}\label{thm:binary-to-many} For every $\star$-graph $\structure{H}_\star$, there exists a $(\symb{\star},\sigma)$-structure $\structure{A}_\star$ such that the problems
  $\MP_\star(\structure{H}_\star)$ and $\MP_\star^\sigma(\structure{A}_\star)$ are $\polynomial$-time equivalent.
\end{theorem}

\begin{proof}
Let us recall from Section~\ref{sec:fromemptysettocsp} that there is a one-to-one correspondence between  $(\symb{\varnothing},\sigma)$-structures and relational $\sigma_\CSP$-structures such that for every two $(\symb{\varnothing},\sigma)$-structures $\structure{A}_\varnothing,\structure{B}_\varnothing$ and the corresponding $\sigma_\CSP$-structures $\structure{A}_\CSP,\structure{B}_\CSP$:
\begin{displaymath}
\structure{B}_\varnothing\to\structure{A}_\varnothing\Leftrightarrow\structure{B}_\CSP\to\structure{A}_\CSP.
\end{displaymath}
Now, let us consider a $\star$-graph $\structure{H}_\star$ with its corresponding relational $\gamma_\CSP$-structure $\structure{H}_\CSP$. The signature $\sigma_\CSP$ consists of relation symbols $\rel{R}_{i,j}$, for $i\in[n]$ and $j\in\{0,1\}$. Every relation $\rel{R}_{i,j}^{\structure{A}_\CSP}$ is primitively positively defined with a corresponding formula $\varphi_{i,j}$, where
\begin{equation}\label{eq:bi-many1}
\text{for each } j\in\{0,1\}\colon \varphi_{1,j}(x_1,\ldots,x_{k_1}) := (E_j(x_1,x_2)=1)\wedge(x_3=x_3)\wedge\dots\wedge(x_{k_1}=x_{k_1}),
\end{equation}
\begin{equation}\label{eq:bi-many2}
\text{for each } i>1\text{ and }j\in\{0,1\}\colon \varphi_{i,j}(x_1,\ldots,x_{k_i}) := (x_1=x_1)\wedge\dots\wedge(x_{k_i}=x_{k_i}). 
\end{equation}
Observe that the relational $\gamma_\CSP$-structure $\structure{H}_\CSP$ is also pp-definable from the relational $\sigma_\CSP$-structure $\structure{A}_\CSP$:
\begin{equation}\label{eq:many-bi}
\text{for each }j\in\{0,1\},\; \rel{E}_j(x_1,x_2) = 1 \Leftrightarrow \exists x_3,\ldots,x_{k_1}\; \rel{R}_{1,j}(x_1,x_2,\ldots,x_{k_1}).
\end{equation}
Now, consider a $\star$-graph $\structure{G}_\star$. Since every $\star$-graph is also a $\varnothing$-graph, there is a relational $\gamma_\CSP$-structure
$\structure{G}_\CSP$ such that $\structure{G}_\star\to\structure{H}_\star$ if and only if $\structure{G}_\CSP\to\structure{H}_\CSP$. By the pp-definability in \cref{eq:many-bi} and by
Theorem~\ref{thm:pp-definition}, we can construct a relational $\sigma_\CSP$-structure $\structure{B}_\CSP$ such that $\structure{G}_\CSP\to\structure{H}_\CSP$ if
and only if $\structure{B}_\CSP\to\structure{A}_\CSP$. From $\structure{B}_\CSP$, we obtain a $(\symb{\star},\sigma)$-structure
$\structure{B}_\star$ such that $\structure{B}_\CSP\to\structure{A}_\CSP$ if and only if $\structure{B}_\star\to\structure{A}_\star$, by Corollary~\ref{cor:mptocsp}. Notice
that, because $\structure{G}_\star$ is a $\star$-graph, for every $(x,y)\in G^2$, we have either $E_0^{\structure{G}_\CSP}(x,y)=1$ or $E_1^{\structure{G}_\CSP}(x,y)=1$. Thus, in $\structure{B}_\CSP$, every relation other than $\rel{R}_1$ is interpreted trivially, and, for each tuple $\tuple{x}\in B^{k_1}$, either $\rel{R}_{1,0}^{\structure{B}_\CSP}(\tuple{x})=1$ or $\rel{R}_{1,1}^{\structure{B}_\CSP}(\tuple{x})=1$. So, $\structure{B}_\star$ is indeed a $(\symb{\star},\sigma)$-structure, that finishes the reduction from $\MP_\star(\structure{H}_\star)$ to $\MP_\star^\sigma(\structure{A}_\star)$.

For the other direction, consider any $(\symb{\star},\sigma)$-structure $\structure{B}_\star$. Similarly, we construct a relational $\sigma_\CSP$-structure $\structure{B}_\CSP$, and by the pp-definition in
\cref{eq:bi-many1,eq:bi-many2}, we can compute a relational $\gamma_\CSP$-structure $\structure{G}_\CSP$ such that
$\structure{G}_\CSP\to \structure{H}_\CSP$ if and only if $\structure{B}_\star\to \structure{A}_\star$, and then a $\star$-graph $\structure{G}_\star$ such that
$\structure{B}_\star\to\structure{A}_\star$ if and only if $\structure{G}_\star\to\structure{H}_\star$. With similar arguments as in the other direction, we can prove
that $\structure{G}_\star$ is indeed a $\star$-graph. We have thus shown that $\MP_\star(\structure{G}_\star)$ and $\MP_\star^\sigma(\structure{A}_\star)$ are $\polynomial$-time
equivalent.
\end{proof}

\begin{remark}
One notices that the proof of Theorem~\ref{thm:binary-to-many} is still correct if we replace $\gamma$ by any relation $\rel{R}$ of arity $\ell \geq 2$, we require
in this case that $\rel{R_1}$ has arity at least $\ell$.  
\end{remark}

\subsection{From many relations to one}\label{subsec:many_to_one}

Let $\sigma = \{\rel{R}_1,\ldots,\rel{R}_p\}$ be a finite relational signature with arities $k_1,\ldots,k_p$, and let $k := \max_{1\leq i\leq p} k_i$. In this subsection, we show that there exists $\tilde{\sigma} = \{\rel{R}\}$ with $\rel{R}$ of arity $k+p-1$, such that, for every $(\symb{\star},\sigma)$-structure $\structure{A}$, there exists a $(\symb{\star},\tilde{\sigma})$-structure $\tilde{\structure{A}}$ such that $\MP_\star^\sigma(\structure{A})$ and $\MP_\star^{\tilde{\sigma}}(\tilde{\structure{A}})$ are $\polynomial$-time equivalent.

Let us first describe how the $\tilde{\sigma}$-structure $\tilde{\structure{A}}$ is constructed from a $\sigma$-structure $\structure{A}$. If $A$ is the domain of $\structure{A}$, then the domain  of $\tilde{\structure{A}}$ is $\tilde{A} := A \sqcup \{c_A\}$, with  a new element $c_A$. First, define two important types of tuples $\family{A}_1$ and $\family{A}_2$. The first one contains the information about the tuples of $\structure{A}$, and the second one consists of one special tuple.

\begin{equation*}
\family{A}_1 := \{\tilde{\tuple{t}}=(\underbrace{c_A,\ldots,c_A}_{i-1},\tuple{t},\underbrace{c_A,\ldots,c_A}_{k+p-k_i-i})\mid\rel{R}_i\in\sigma,\tuple{t}\in A^{k_i}\}, \; \family{A}_2 := \{(\underbrace{c_A,\ldots,c_A}_{k+p-1})\}
\end{equation*}

The relation $\rel{R}^{\tilde{\structure{A}}}$ is defined as follows:
\begin{equation}\label{eq:frombtotildeb}
\rel{R}^{\tilde{\structure{A}}}(\tilde{\tuple{t}}) = \begin{cases}\rel{R}^\structure{A}_i(\tuple{t}) & \text{if }\tilde{\tuple{t}}\in\family{A}_1,\\ 1 & \text{if } \tilde{\tuple{t}}\in\family{A}_2,\\ 0 & \text{otherwise.}\end{cases}
\end{equation}

Now we will prove one direction of the $\polynomial$-time equivalence. The size of $\tilde{\structure{A}}$ is polynomial in $|A|$, so the construction takes $\polynomial$-time, and below we show that $\structure{B}\to\structure{A}\Leftrightarrow\tilde{\structure{B}}\to\tilde{\structure{A}}.$

\begin{lemma}\label{lem:sigmatotildesigma}
$\MP_\star^\sigma(\structure{A})$ reduces in polynomial time to $\MP_\star^{\tilde{\sigma}}(\tilde{\structure{A}})$.
\end{lemma}
\begin{proof}
Let $\structure{B}$ be an input instance of the problem $\MP_\star^{\sigma}(\structure{A})$. Suppose that there is $h\colon \structure{B}\to \structure{A}$ -- a homomorphism. We will show that $\tilde{h}\colon\tilde{B}\to\tilde{A}$ is a homomorphism, where
\begin{equation*}
    \tilde{h}(x) = \begin{cases}c_A &\text{if }x=c_B,\\ h(x)&\text{otherwise.}\end{cases}
\end{equation*}
Let $\family{B}_1,\family{B}_2$ be defined in a similar way to $\family{A}_1,\family{A}_2$. Consider $\tilde{\tuple{t}}=(c_B,\ldots,c_B,\tuple{t},c_B,\ldots,c_B)\in\family{B}_1$, where $\tuple{t} = (b_1,\ldots,b_{k_i})\in B^{k_i}$ and $k_i$ is the arity of $\rel{R}_i$ in $\sigma$. Then, the tuple $\tilde h(\tilde{\tuple{t}})=(c_A,\ldots,c_A,h(\tuple{t}),c_A,\ldots,c_A)$ is in $\family{A}_1$. As $h$ is a homomorphism, we have, by \cref{eq:frombtotildeb}, that:
\begin{displaymath}
\rel{R}^{\tilde{\structure{B}}}(\tilde{\tuple{t}}) = \rel{R}_i^\structure{B}(\tuple{t})\preceq_\star \rel{R}_i^\structure{A}(h(\tuple{t})) = \rel{R}^{\tilde{\structure{A}}}(\tilde h(\tilde{\tuple{t}}))
\end{displaymath}
For $\tilde{\tuple{t}}\in\family{B}_2$, we have that $\tilde h(\tilde{\tuple{t}})=(c_A,\ldots,c_A)$, so $\rel{R}^{\tilde{\structure{A}}}(\tilde h(\tilde{\tuple{t}}))=\rel{R}^{\tilde{\structure{B}}}(\tilde{\tuple{t}})=1.$
Let us consider a tuple $\tilde{\tuple{t}}=(x_1,\ldots,x_{k+p-1})\not\in\family{B}_1\sqcup\family{B}_2$. We know that $\tilde{h}(x)=c_A$ if and only if $x=c_B$, thus $\tilde h(\tilde{\tuple{t}})\not\in\family{A}_1\sqcup\family{A}_2$. Then $\rel{R}^{\tilde{\structure{A}}}(\tilde h(\tilde{\tuple{t}}))=\rel{R}^{\tilde{\structure{B}}}(\tilde{\tuple{t}})=0.$ We have shown that $\tilde{h}$ is a homomorphism.

Suppose now that there is a homomorphism $\tilde{h}\colon\tilde{\structure{B}}\to\tilde{\structure{A}}$. We know that $x=c_B$ if and only if
$\rel{R}^{\tilde{\structure{B}}}(x,\ldots,x)=1$, and otherwise $\rel{R}^{\tilde{\structure{B}}}(x,\ldots,x)=0$. A similar thing holds for $\tilde{\structure{A}}$. Thus, $x=c_B$ if and only if $\tilde{h}(x)=c_A$. This allows us to correctly construct $h\colon \structure{B}\to\structure{A}$, where,
for all $x\in B$, $h(x)=\tilde{h}(x)$.

For each $\rel{R}_i\in\sigma$ and $\tuple{t}\in B^{k_i}$, $\tuple{t}$ is associated with $\tilde{\tuple{t}}=(c_B,\ldots,c_B,\tuple{t},c_B\ldots,c_B)\in \family{B}_1$
and its image $h(\tuple{t})\in A^{k_i}$ is associated with $\tilde{h}(\tilde{\tuple{t}}) = (c_A,\ldots,c_A,h(\tuple{t}),c_A,\ldots,c_A)\in\family{A}_1$. We know, by the construction of $\tilde{\structure{A}}$ and $\tilde{\structure{B}}$, and by \cref{eq:frombtotildeb}, that:
\begin{displaymath}
\rel{R}_i^\structure{B}(\tuple{t}) = \rel{R}^{\tilde{\structure{B}}}(\tilde{\tuple{t}}) \preceq_\star \rel{R}^{\tilde{\structure{A}}}(\tilde h(\tilde{\tuple{t}})) = \rel{R}_i^\structure{A}(h(\tuple{t})).
\end{displaymath}
So, $h$ is a homomorphism and $\MP_\star^\sigma(\structure{A})$ reduces to $\MP_\star^{\tilde{\sigma}}(\tilde{\structure{A}})$.
\end{proof}

Now we have to find in polynomial time, for any input $(\symb{\star},\tilde{\sigma})$-structure $\tilde{\structure{G}}$ of $\MP_\star^{\tilde{\sigma}}(\structure{A})$, a $(\symb{\star},\sigma)$-structure $\structure{B}$ such that 
\begin{displaymath}
\tilde{\structure{G}}\to\tilde{\structure{A}}\Leftrightarrow\structure{B}\to\structure{A}.
\end{displaymath}

\begin{lemma}\label{lem:tildesigmatosigma}
$\MP_\star^{\tilde{\sigma}}(\tilde{\structure{A}})$ reduces in polynomial time to $\MP_\star^\sigma(\structure{A})$.
\end{lemma}
\begin{proof}
Let $\tilde{\structure{G}}$ be an input instance of $\MP_\star^{\tilde{\sigma}}(\tilde{\structure{A}})$. Firstly, for every element $x\in\tilde{G}$, we check
whether $\rel{R}^{\tilde{\structure{G}}}(x,\ldots,x)=\star$. If such an $x$ exists, then we cannot map $\tilde{\structure{G}}$ to $\tilde{\structure{A}}$
as, for all $y\in\tilde{A}$, we have that $\rel{R}^{\tilde{\structure{A}}}(y,\ldots,y)\in\{0,1\}$. This can be checked in time linear in $|\tilde{G}|$. In
this case, we output some fixed NO input instance of $\MP_\star^\sigma(\structure{A})$, e.g., some $\structure{B}$ where there is $b\in B$ and
$\rel{R}_i^\structure{B}(b,\ldots,b)=\star$ for all $\rel{R}_i\in \sigma$.

Now we can assume that, for all $x\in\tilde{G}$, $\rel{R}^{\tilde{\structure{G}}}(x,\ldots,x)\in\{0,1\}$.
We partition the elements of $\tilde{G}$ into two sets $\{C_0,C_1\}$, where, for each $x\in \tilde{G}$,
\begin{equation}\label{eq:definingc1}
x\in C_i \Leftrightarrow \rel{R}^{\tilde{\structure{G}}}(x,\ldots,x) = i.
\end{equation}
As, for all $a\in \tilde{A}\setminus \{c_A\}$, $\rel{R}^{\tilde{\structure{A}}}(a,\ldots,a)=0$, the existence of a homomorphism $h\colon \tilde{\structure{G}}\to\tilde{\structure{A}}$ implies that, for all $x\in \tilde{G}$, we have $h(x) = c_A \Leftrightarrow x\in C_1$. We are going to construct a $\tilde{\sigma}$-structure $\tilde{\structure{B}}$ with the following properties:
\begin{enumerate}
        \item $\tilde{\structure{G}}\to\tilde{\structure{B}}$;
        \item $\tilde{\structure{G}}\to\tilde{\structure{A}}\Leftrightarrow\tilde{\structure{B}}\to\tilde{\structure{A}}$;
        \item Either we can check in $\polynomial$-time that $\tilde{\structure{B}}\not\to\tilde{\structure{A}}$ or there exists a $\sigma$-structure $\structure{B}$ such that $\tilde{\structure{B}}$ is obtained from $\structure{B}$ by the construction from the beginning of Subsection~\ref{subsec:many_to_one}.
\end{enumerate}

The domain $\tilde{B}$ is $C_0\sqcup\{c_B\}$. The element $c_B$ should be considered as the result of identifying all the elements in $C_1$ into a single element, namely $c_B$.
    
Let us consider a tuple $\tilde{\tuple{t}}=(b_1,\ldots,b_{k+p-1})\in\tilde{B}^{k+p-1}$. Denote by $\family{I}_{\tilde{\tuple{t}}}\subseteq[k+p-1]$ the set of indices such that $b_i=c_B$. Denote by $\family{C}_{\tilde{\tuple{t}}}$ the set of all tuples $(x_1,\ldots,x_{k+p-1})\in\tilde{G}^{k+p-1}$ such that, for each $i\in[k+p-1]$, 
\begin{equation*}
\left(i\in\family{I}_{\tilde{\tuple{t}}}\Rightarrow x_i \in C_1\right) \wedge \left(i\notin\family{I}_{\tilde{\tuple{t}}}\Rightarrow b_i=x_i\right). 
\end{equation*}
The interpretation $\rel{R}^{\tilde{\structure{B}}}$ is defined as follows, here $\bigvee$ denotes the join operation w.r.t. $\preceq_{\star}$:
\begin{equation}\label{eq:identifyingc1}
\rel{R}^{\tilde{\structure{B}}}(\tilde{\tuple{t}}) = \bigvee_{(x_1,\ldots,x_{k+p-1})\in\family{C}_{\tilde{\tuple{t}}}}\rel{R}^{\tilde{\structure{G}}}(x_1,\ldots,x_{k+p-1}).
\end{equation}
Observe that we can construct $\tilde{\structure{B}}$ in time polynomial in the size of the input $\tilde{\structure{G}}$.
    
Let us check the property 1: $\tilde{\structure{G}}\to\tilde{\structure{B}}$. Let $\pi\colon\tilde{G}\to\tilde{B}$ be a mapping such that
\begin{equation*}
\pi(x) = \begin{cases}c_B & \text{if } x\in C_1,\\ x & \text{if } x\in C_0.\end{cases}
\end{equation*}
Let $\tilde{\tuple{x}}=(x_1,\ldots,x_{k+p-1})$ be a tuple in $\tilde{G}^{k+p-1}$.
As $\tilde{\tuple{x}}\in\family{C}_{\pi(\tilde{\tuple{x}})}$, we have, by \cref{eq:identifyingc1}, that $\rel{R}^{\tilde{\structure{G}}}(\tilde{\tuple{x}})\preceq_\star \rel{R}^{\tilde{\structure{B}}}(\pi(\tilde{\tuple{x}}))$. This proves that $\pi$ is a homomorphism.

Let us check the property 2: $\tilde{\structure{G}}\to\tilde{\structure{A}}\Leftrightarrow\tilde{\structure{B}}\to\tilde{\structure{A}}$. Since $\tilde{\structure{G}}\to\tilde{\structure{B}}$ (property 1), we need to show only the ``$\Rightarrow$'' direction.  Assume that there is
$h_G\colon\tilde{\structure{G}}\to\tilde{\structure{A}}$ -- a homomorphism. Notice that, for all $x\in G$, we have that  $x\in C_1\Leftrightarrow h_G(x)=c_A$. We define a mapping $h_B\colon \tilde{B}\to\tilde{A}$ such that
\begin{equation*}
h_B(x) = \begin{cases}c_A & \text{if } x=c_B\\ h_G(x) & \text{otherwise.}\end{cases}
\end{equation*}
Let $\tilde{\tuple{t}}=(b_1,\ldots,b_{k+p-1})$ be a tuple in $\tilde{B}^{k+p-1}$. 
Observe that, for each tuple $\tuple{x}\in\family{C}_{\tilde{\tuple{t}}}$, we have $h_B(\tilde{\tuple{t}})=h_G(\tuple{x})$.
We also know that $\rel{R}^{\tilde{\structure{A}}}(h_B(\tilde{\tuple{t}}))\succeq_\star \rel{R}^{\tilde{\structure{G}}}(x_1,\ldots,x_{k+p-1})$ for all $(x_1,\ldots,x_{k+p-1})\in\family{C}_{\tilde{\tuple{t}}}$. Thus, 
\begin{displaymath}
\rel{R}^{\tilde{\structure{A}}}(h_B(\tilde{\tuple{t}}))\succeq_\star \bigvee_{(x_1,\ldots,x_{k+p-1})\in \family{C}_{\tilde{\tuple{t}}}}\rel{R}^{\tilde{\structure{G}}}(x_1,\ldots,x_{k+p-1}) = \rel{R}^{\tilde{\structure{B}}}(\tilde{\tuple{t}}).
\end{displaymath}
This shows that $h_B$ is a homomorphism.

Finally, we need to check the property 3 to finish the proof. Recall that we split all the tuples $(b_1,\ldots,b_{k+p-1})\in\tilde{B}^{k+p-1}$ into three
classes: $\family{B}_1,\family{B}_2$, and the rest. Observe that, for each homomorphism $h:\tilde{\structure{B}}\to
\tilde{\structure{A}}$ and each $x$ in $\tilde{B}$, we have that $x=c_B$ if and only if $h(x)=c_A$. In particular,
  $h(\family{B}_1)\subseteq \family{A}_1$ and $h(\family{B}_2)\subseteq \family{A}_2$. We first look at the tuple $\tilde{\tuple{t}}=(c_B,\ldots,c_B)\in\family{B}_2$. By \cref{eq:definingc1,eq:identifyingc1}, we know that $\rel{R}^{\tilde{\structure{B}}}(\tilde{\tuple{t}})\succeq_\star 1$. If $\rel{R}^{\tilde{\structure{B}}}(\tilde{\tuple{t}})=\star$, then there is no homomorphism from $\tilde{\structure{B}}$ to $\tilde{\structure{A}}$, so we output some fixed NO input instance of $\MP_\star^\sigma(\structure{A})$ for $\tilde{\structure{G}}$. If $\rel{R}^{\tilde{\structure{B}}}(\tilde{\tuple{t}})=1$, then we continue. 
    
If there exists $\tilde{\tuple{t}}\not\in\family{B}_1\sqcup\family{B}_2$ such that $\rel{R}^{\tilde{\structure{B}}}(\tilde{\tuple{t}})\not=0$, then there is no homomorphism from $\tilde{\structure{B}}$ to $\tilde{\structure{A}}$, so we output some fixed NO input instance of $\MP_\star^\sigma(\structure{A})$ for $\tilde{\structure{G}}$. If, for all tuples $\tilde{\tuple{t}}\not\in\family{B}_1\sqcup\family{B}_2$, we have that $\rel{R}^{\tilde{\structure{B}}}(\tilde{\tuple{t}})=0$, then we continue. We can do all these checks in time polynomial in $|\tilde{G}|$.
    
We can now assume that $\rel{R}^{\tilde{\structure{B}}}(c_B,\dots,c_B)=1$ and that, for $\tilde{\tuple{t}}\not\in\family{B}_1\sqcup\family{B}_2$, $\rel{R}^{\tilde{\structure{B}}}(\tilde{\tuple{t}})=0$. We are ready to construct the $(\symb{\star},\sigma)$-structure $\structure{B}$:
\begin{itemize}
        \item the domain $B$ of $\structure{B}$ is $\tilde{B}\setminus\{c_B\}$;
        \item for each relation $\rel{R}_i\in\sigma$ and each tuple $\tuple{t}=(b_1,\ldots,b_{k_i})\in B^{k_i}$, let
        \begin{equation*}
        \rel{R}^{\structure{B}}_i(\tuple{t})=\rel{R}^{\tilde{\structure{B}}}(\underbrace{c_B,\ldots,c_B}_{i-1},\tuple{t},\underbrace{c_B,\ldots,c_B}_{k+p-k_i-i}).
        \end{equation*}
\end{itemize}

It can be easily checked that $\tilde{\structure{B}}$ is obtained from $\structure{B}$ by the construction from the beginning of Subsection~\ref{subsec:many_to_one}. By Lemma~\ref{lem:sigmatotildesigma}, $\structure{B}\to\structure{A}$ if and only if $\tilde{\structure{B}}\to\tilde{\structure{A}}$.  We have shown that, for
  each $(\symb{\star},\tilde{\sigma})$-structure $\tilde{\structure{G}}$, we can find in time polynomial in $|\tilde{G}|$ a
  $(\symb{\star},\sigma)$-structure $\structure{B}$ such that
  $\tilde{\structure{G}}\to\tilde{\structure{A}}\Leftrightarrow \tilde{\structure{B}}\to\tilde{\structure{A}}$.  Thus,
  $\MP_\star^{\tilde{\sigma}}(\tilde{\structure{A}})$ reduces in polynomial time to $\MP_\star^\sigma(\structure{A})$.
\end{proof}
Lemma~\ref{lem:sigmatotildesigma} and Lemma~\ref{lem:tildesigmatosigma} provide the following statement about the dichotomy property.
\begin{theorem}\label{thm:manytoone}
If the class of problems $\MP_\star^{\tilde{\sigma}}$ has a dichotomy, then the class $\MP_\star^{\sigma}$ has a dichotomy.
\end{theorem}

Our results from Sections~\ref{sec:mpstar=mp} and~\ref{sec:arity} are summarised on Figure~\ref{fig:arities_diagram}. One can see now that the existence of a dichotomy for $\MP^{\tilde{\sigma}}$ implies a similar dichotomy for all other classes considered on the figure.
Observe that in order to prove the other direction, for every $(\symb{\star},\tilde{\sigma})$-structure $\structure{A}$, we have to find a
$(\symb{\star},\sigma)$-structure $\hat{\structure{A}}$ such that $\MP_\star^{\sigma}(\hat{\structure{A}})$ and $\MP_\star^{\tilde{\sigma}}(\structure{A})$
are $\polynomial$-time equivalent. We discuss in the next subsection why this problem is difficult.

\begin{figure}
\begin{center}
\begin{tikzcd}
\MP_\star \arrow[d, leftrightarrow] & \MP_\star^{\tilde{\sigma}} \arrow[d, leftrightarrow] \arrow[r] & \MP_\star^\sigma \arrow[ll, bend right] \arrow[d, leftrightarrow]\\
\MP & \MP^{\tilde{\sigma}} & \MP^\sigma
\end{tikzcd}
\end{center} 
\caption{Dichotomy implications. Each arrow shows an implication of the existence of a dichotomy, i.e., if the class at the tail has a dichotomy, then the class at the head has it. The vertical ones are shown in \Cref{sec:mpstar=mp}, and the horizontal ones are shown in \Cref{sec:arity}.}
\label{fig:arities_diagram}
\end{figure}

\subsection{From one relation to directed graphs}\label{subsec:onetobinary}

Before the dichotomy question for finite $\CSP$ was solved, it had already been known that the choice of the relational signature did not matter. This was implied by the result of Feder and Vardi~\cite{federvardi1998} who showed that, for every finite $\CSP$, there is a $\polynomial$-time equivalent $\CSP$ on directed graphs. For more details about the reductions, see Bulin et al.~\cite{bulin2015}.

The transformation consisted of two steps. At first, a $\CSP$ over some signature $\sigma$ was transformed to an equivalent $\CSP$ over a signature $\tilde{\sigma}$ with a single relation symbol of some arity possibly greater than 2.  Secondly, the interim $\CSP^{\tilde{\sigma}}$ was transformed to an equivalent $\CSP$ on directed graphs. The resulting digraph was obtained by replacing relational tuples with copies of some ``gadget'' digraph.

Notice that, regarding the first step, we showed a similar result for $\MP$ in Subsection~\ref{subsec:many_to_one}. We do not prove the second step in this subsection. However, we generalise the transformation from~\cite{federvardi1998,bulin2015} by providing a list of necessary conditions that describe this type of transformations. We will show that every transformation that satisfies these conditions provides a reduction from an $\MP_\star$ problem to another $\MP_\star$ problem that is $\polynomial$-time equivalent to some $\MP_\varnothing$ problem. As $\MP_\varnothing$ has a dichotomy, by Corollary~\ref{cor:mptocsp}, the possible existence of a backwards reduction would imply a dichotomy for $\MP_\star$ over any finite relational signature.

To simplify the notation, we consider a signature with a ternary relation symbol: $\tilde{\sigma} = \{\rel{R}(\cdot,\cdot,\cdot)\}$. Let $\structure{H}_3$ be some $(\star,\tilde{\sigma})$-structure. Below, we formally describe how should a transformation look like. At first, we explain how to obtain a $\star$-graph $\structure{H}_2$ from $\structure{H}_3$. After that, we give a list of conditions that we impose on the $\star$-graph $\structure{H}_2$ and on its ``gadgets''. Finally, we state Proposition~\ref{prop:twotoincomplete} that shows why this exact type of transformations does not work well in the case of Matrix Partitions. We put the proof into the appendix for the interested reader.

The $\star$-graph $\structure{H}_2$ is constructed from a $(\star,\tilde{\sigma})$-structure $\structure{H}_3$ as follows.
\begin{enumerate}
    \item First take the same domain $H_2:=H_3$.
    \item Then, for every value $v\in\{0,1,\star\}$, define a \emph{gadget} $\star$-graph $\structure{T}^v$.
    \item Then, for every tuple $\tuple{t} = (x_1,x_2,x_3)$ such that $\rel{R}^{\structure{H}_3}(\tuple{t})=v$, replace it with a copy of $\structure{T}^v$, denoted by $\structure{T}_\tuple{t}^v$,  such that this copy contains $x_1,x_2,x_3$, while its other elements are newly introduced.
\end{enumerate}

In the $\CSP$ case~\cite{federvardi1998,bulin2015}, every such gadget was a digraph obtained from a rooted tree with three leaves, where the leaves were the elements $x_1,x_2,x_3$ of $H_3$. During the reduction from $\CSP$ on directed graphs to $\CSP^{\tilde{\sigma}}$, it was clear which elements of the input directed graph
must be the elements of the domain of the $\tilde{\sigma}$-structure to which this directed graph was reduced. We generalise this constructive approach by the
conditions imposed on $\structure{H}_2$ and on the gadgets $\structure{T}^v$. At first, we list the conditions for the target structure $\structure{H}_2$.
\begin{enumerate}
    \item \emph{The $H_3$-part is preserved by homomorphisms.} Let $\structure{H}_2,\structure{H}_2'$ be two $\star$-graphs obtained from $(\symb{\star},\tilde{\sigma})$-structures
      $\structure{H}_3,\structure{H}_3'$. Then, any homomorphism $h\colon \structure{H}_2\to\structure{H}_2'$ maps the subset $H_3$ to $H_3'$, i.e., for all $x$ in $H_2$, we have that $x\in H_3 \Leftrightarrow h(x)\in H_3'$. 
    \item \emph{It is easy to single out the $H_3$-part of the input.} For any $\star$-graph $\structure{G}$, one of the two following statements must hold.
    \begin{itemize}
        \item One can decide in $\polynomial$-time in $|G|$ if $\structure{G}\to\structure{H}_2$.
        \item For every $x\in G$ and every two homomorphisms $h_1,h_2\colon \structure{G}\to\structure{H}_2$, we have that $h_1(x)\in H_3$ if and only if $h_2(x)\in H_3$. Moreover, for every $x\in G$, one can decide in $\polynomial$-time in $|G|$ if, for every homomorphism $h\colon \structure{G}\to\structure{H}$, we have that $h(x)$ belongs to $H_3$.
    \end{itemize} 
\end{enumerate}

Next, we list the conditions for the gadgets $\structure{T}^v$ of $\structure{H}_2$.

\begin{enumerate}
\item[3.] \emph{Gadgets have tractable $\CSP$s and respect the partial order.} For each $v\in\{0,1,\star\}$, the problem $\MP_\star(\structure{T}^v)$ is solvable in $\polynomial$-time and, for each $v,v'\in\{0,1,\star\}$, we have that $v\preceq_\star v'$ if and only if $\structure{T}^{v}\to\structure{T}^{v'}$.
\item[4.] \emph{Copies of gadgets touch only by elements of $H_3$.} For two elements $w,w'\in H_2$ such that $w,w'\notin H_3$ and $w,w'$ do not belong to the same gadget $\structure{T}^v_{xyz}$, we have that $\rel{E}^{\structure{H}_2}(w,w') = 0$.
\item[5.] \emph{For connected inputs, the coordinates of $H_3$-elements are uniquely determined.} Call a $\star$-graph $\structure{A}$ \emph{connected} if, for every $a,a'\in A$, there is a sequence of elements $a_1,\ldots,a_n$ such that $a=a_1,a'=a_n$, and, for $i\in[n-1]$, one of $\rel{E}^\structure{A}(a_i,a_{i+1}),\rel{E}^\structure{A}(a_{i+1},a_i)$ belongs to $\{1,\star\}$. Then, for every connected $\star$-graph $\structure{A}$, one of the two following statements must hold.
\begin{itemize}
    \item For every $v\in \{0,1,\star\}$, there is a homomorphism $h\colon \structure{A}\to\structure{T}^v_{xyz}$.
    \item For every $v\in\{0,1,\star\}$ and every two homomorphisms $h,h'\colon \structure{A}\to\structure{T}_{xyz}^v$ and every $a\in A$ such that $h(a)\in\{x,y,z\}$, we have that $h(a)=h'(a)$.
\end{itemize}
\end{enumerate}

The transformation from~\cite{federvardi1998,bulin2015} satisfies all the five conditions. However, the third condition is not really applicable for the $\CSP$ case because there is only one gadget type. In the following proposition we state that, for every construction of $\structure{H}_2$ that satisfies these five conditions, proving that $\MP_\star^{\tilde{\sigma}}(\structure{H}_3)$ and $\MP_\star(\structure{H}_2)$ are $\polynomial$-time equivalent is at least as hard as proving the existence of a dichotomy for $\MP_\star^{\tilde{\sigma}}$, as, by Corollary~\ref{cor:mptocsp}, $\MP_\varnothing^{\tilde{\sigma}}$ has a dichotomy.

\begin{proposition}\label{prop:twotoincomplete}
  Let a $\star$-graph $\structure{H}_2$ be constructed from some $(\star,\tilde{\sigma})$-structure $\structure{H}_3$ and satisfy all the five
  conditions above. Then, $\MP_\varnothing^{\tilde{\sigma}}(\structure{H}_3)$ and $\MP_\star(\structure{H}_2)$ are $\polynomial$-time equivalent.
\end{proposition}

\section{Obstructions}\label{sec:obstructions}

Throughout this section,  assume that $\sigma$ is a fixed finite relational signature and that, for $*\in\{01,\star,\varnothing\}$, $\Cat_*$
denotes the set of all $(\symb{*},\sigma)$-structures.
The following definition extends the notion of obstructions~\cite{federhell2007} to $(\symb{*},\sigma)$-structures.
\begin{definition}[Obstruction set]
Let $* \in \{01,\star,\varnothing\}$ and let $\structure{H}$ be a $\star$-structure.  A $*$-structure
  $\structure{G}$ is called a \emph{minimal obstruction} for $\MP_*(\structure{H})$ if $\structure{G}\not\to\structure{H}$ and, for all $v\in G$,
  $\structure{G}\setminus\{v\}\to\structure{H}$. The set of all minimal obstructions for  $\MP_*(\structure{H})$ is denoted by
  $\Obs_*^\subset(\structure{H})$.
\end{definition}

The following extends the notion of finite duality~\cite{atserias2008} to $(\symb{*},\sigma)$-structures.

\begin{definition}[Finite duality]
A set $\family{F}$ of $*$-structures is a \emph{duality set} for the problem $\MP_*(\structure{H})$ if
\begin{align*}
  \structure{G}\in \MP_*(\structure{H})  \Longleftrightarrow \text{for all $\structure{F}\in\family{F}$,}\;\structure{F}\not\to \structure{G}
\end{align*}
If, moreover, the set $\family{F}$ is finite, then $\MP_*(\structure{H})$ has \emph{finite duality}.
\end{definition}

We prove that, for each of the three problems $\MP(\structure{H}),\MP_\star(\structure{H}),\MP_\varnothing(\structure{H})$, having a finite duality and having a finite set of minimal obstructions is the same, that $\MP(\structure{H})$ and $\MP_\star(\structure{H})$ agree on having this property, and that the existence of such a characterisation for $\MP_\varnothing(\structure{H})$ yields the same result for both cases $01$ and $\star$. Our results are depicted on the diagram below, where each arrow stands for a logical implication.

\begin{center}
\begin{tikzcd}\label{fig:obs_diagram}

\vert \Obs_{01}^\subset(\structure{H})\vert <\infty \arrow[d, "\textrm{Cor.}~\ref{cor:homimpliesinclusion}", leftrightarrow] \arrow[r,"\textrm{Prop.}\ref{prop:obs_01_obs_star}", leftrightarrow]  &  
\vert \Obs_\star^\subset(\structure{H})\vert <\infty \arrow[d, "\textrm{Cor.}\ref{cor:obs_star_fd}", leftrightarrow]  & 
\vert \Obs_\varnothing^\subset(\structure{H})\vert <\infty  \arrow[d,"\textrm{Prop.}\ref{prop:empty_fd_obs}\textrm{,}\ref{prop:empty_obs_fd}", leftrightarrow] \arrow[l,"\textrm{Prop.}\ref{prop:obs_empty_obs_star}"]\\
\MP(\structure{H})\text{ has f.d.} \arrow[r, "\textrm{Cor.}\ref{cor:hommin-star-01}", leftrightarrow] & 
\MP_\star(\structure{H})\text{ has f.d.} & 
\MP_\varnothing(\structure{H})\text{ has f.d.}\arrow[l,"\textrm{Prop.}\ref{prop:mp_empty_mp_star}"]

\end{tikzcd}
\end{center}

\subsection{Looking at $\Cat_{01}$ and at $\Cat_{\star}$}

Let us first prove that, in $\Cat_{01}$, the minimal obstruction set is also a duality set.

\begin{proposition}\label{prop:homimpliesinclusion} 
$\Obs_{01}^{\subset}(\structure{H})$ is a duality set for $\MP(\structure{H})$. Moreover, among all duality sets,  $\Obs_{01}^{\subset}(\structure{H})$ is the minimal one by inclusion.
\end{proposition}

\begin{proof}
Let $\structure{A}$ be in $\Cat_{01}$ such that $\structure{A}\not\to\structure{H}$. We start by iteratively removing arbitrary elements from
  $\structure{A}$ until the substructure induced by remaining elements is a minimal obstruction, i.e., if we remove any element, then the resulting structure will map to $\structure{H}$. Such a substructure belongs to $\Obs_{01}^\subset(\structure{H})$ and maps to $\structure{A}$, thus, $\Obs_{01}^\subset(\structure{H})$ is a duality set.

Let $\family{F}$ be a duality set for $\MP(\structure{H})$ such that $|\family{F}|\leq|\Obs_{01}^\subset(\structure{H})|$. We can assume without loss of generality that $\family{F}\subseteq \Obs_{01}^\subset(\structure{H})$: any $\structure{F}$ in $\family{F}$ has an induced substructure that belongs to $\Obs_{01}^\subset(\structure{H})$, so we can substitute this substructure for $\structure{F}$.
Observe that, by definition, all structures of $\Obs_{01}^\subset(\structure{H})$ are cores. Let $\structure{G}$ be in $\Obs_{01}^\subset(\structure{H})\setminus\family{F}$. As $\structure{G}\not\to\structure{H}$, there exists $\structure{G}_1$ in $\family{F}$ such that there is a homomorphism $h\colon \structure{G}_1\to \structure{G}$.
Let $G'= h(G_1)$ and let $\structure{G'}$ be the substructure of $\structure{G}$ induced by $G'$. If $\structure{G'}$ is a proper induced  substructure of $\structure{G}$, then by the assumption of minimality by inclusion, and by transitivity of homomorphism, $\structure{G}_1\to \structure{H}$ -- a  contradiction. Thus, $h(G_1) = G$, but since $h$ is a full homomorphism, $\structure{G}$ is either a proper induced substructure of $\structure{G}_1$ or isomorphic to it. The first one is impossible because $\structure{G}_1$ is a core. Thus, $\structure{G}_1$ is isomorphic to $\structure{G}$ which implies that $\structure{G}\in\family{F}$, a contradiction.
\end{proof}
\begin{corollary}\label{cor:homimpliesinclusion}
$|\Obs_{01}^{\subset}(\structure{H})|<\infty$ if and only if $\MP(\structure{H})$ has finite duality.
\end{corollary}

Observe that the result of Proposition~\ref{prop:homimpliesinclusion} does not hold in $\Cat_\star$.

\begin{proposition}\label{prop:obs-star1} 
For every $01$-structure $\structure{H}$, $\Obs_\star^{\subset}(\structure{H})$ is not a duality set of minimal size. 
\end{proposition}

\begin{proof} We give the proof for $\star$-graphs, the proof for arbitrary signatures is similar. Choose some vertex $x$ from the domain $H$ of $\structure{H}$. Consider a $\star$-graph $\structure{G} = (\{u,v\}, \rel{E}^\structure{G})$ with
  $\rel{E}^\structure{G}(u,u)=\rel{E}^\structure{G}(v,v) = \rel{E}^{\structure{H}}(x,x)$ and
  $\rel{E}^\structure{G}(u,v) = \rel{E}^\structure{G}(v,u) = \star$. Also consider a $\star$-graph $\structure{G}'$ obtained from
  $\structure{G}$ by setting $\rel{E}^{\structure{G'}}(v,u)=0$, and keeping the rest as in $\structure{G}$.

Both $\structure{G}$ and $\structure{G}'$ belong  to  $\Obs_\star^{\subset}(\structure{H})$  as $\structure{H}$ has an element $x$ such that $\rel{E}^\structure{H}(x,x)=\rel{E}^\structure{G}(u,u)=\rel{E}^\structure{G}(v,v)$ and as $\structure{G}\not\to \structure{H}$ and similarly $\structure{G'}\not\to \structure{H}$ because they both  have a $\star$-arc and $\structure{H}$ is a  $01$-graph. Also, $\structure{G}'\to \structure{G}$ and $\structure{G}\not\to \structure{G'}$, so $\structure{G}$ can be removed from $\Obs_\star^{\subset}(\structure{H})$ if it is a duality set.
\end{proof}

\begin{proposition}\label{prop:hommin-star-01}
A family of $01$-structures $\family{F}$ is a duality set for $\MP_\star(\structure{H})$ if and only if $\family{F}$ is a duality set for $\MP(\structure{H})$.
\end{proposition}

\begin{proof}
Let $\family{F}$ be a duality set for $\MP_\star(\structure{H})$. Any $01$-structure $\structure{G}$ is also a $\star$-structure. So, if $\structure{G}\not\to\structure{H}$, then $\structure{F}\to\structure{G}$ for some $\structure{F}$ in $\family{F}$. This means that $\family{F}$ is a duality set for $\MP(\structure{H})$.

Let $\family{F}$ be a finite duality set for $\MP(\structure{H})$. Let $\structure{G}$ be a $\star$-structure that does not map to $\structure{H}$. By Lemma~\ref{lem:mpstardichotomy}, for each $\structure{G}$ in $\Cat_\star\setminus\Cat_{01}$, there exists $\structure{G}_{01}$ in $\Cat_{01}$
such that
  \begin{itemize}
    \item there is a surjective homomorphism $\pi_\structure{G}\colon\structure{G}_{01} \rightarrow \structure{G}$;
    \item $\structure{G}\not\to\structure{G}_{01}$;
    \item $\structure{G}$ is in $\MP_\star(\structure{H})$ if and only if $\structure{G}_{01}$ is in $\MP(\structure{H})$.
  \end{itemize}
As $\structure{G}_{01}$ is in $\Cat_{01}$ and does not map to $\structure{H}$, there exists $\structure{F}$ in $\family{F}$ such that
  $\structure{F}\to\structure{G}_{01}$. By transitivity, we have $\structure{F}\to\structure{G}$. This means that $\family{F}$ is also a duality set for $\MP_\star(\structure{H})$.
\end{proof}

\begin{corollary}\label{cor:hommin-star-01}
$\MP_\star(\structure{H})$ has finite duality if and only if $\MP(\structure{H})$ has finite duality.
\end{corollary}

\begin{remark}
Proposition~\ref{prop:homimpliesinclusion} states that $\Obs_{01}^\subset(\structure{H})$ is the (inclusion-wise) minimal
  duality set for $\MP(\structure{H})$. Proposition~\ref{prop:hommin-star-01} and Corollary~\ref{cor:hommin-star-01}
  imply that $\Obs_{01}^\subset(\structure{H})$ is also the (inclusion-wise) minimal duality set for $\MP_\star(\structure{H})$. So,
  without loss of generality, we can always take $\Obs_{01}^\subset(\structure{H})$ when we consider a duality set for $\MP_\star(\structure{H})$.
\end{remark}

\begin{proposition}\label{prop:obs_01_obs_star}
$\Obs_\star^\subset(\structure{H})$ is finite if and only if $\Obs_{01}^\subset(\structure{H})$ is finite.
\end{proposition}

\begin{proof} Since any $01$-structure is also a $\star$-structure, we can conclude that
  $\Obs_{01}^\subset(\structure{H})\subseteq \Obs_\star^\subset(\structure{H})$, proving the  left-to-right implication.
  
Let us now turn our attention to the other implication and suppose that $\Obs_{01}^\subset(\structure{H})$ is finite.  Let us consider the class $\overline{\Obs_{01}^\subset}(\structure{H})$, that is obtained from  $\Obs_{01}^\subset(\structure{H})$ by taking all $\star$-structures $\structure{A}$ such that there exists a surjective homomorphism  from $\structure{B}$ to $\structure{A}$, for some $\structure{B}$ in $\Obs_{01}^\subset(\structure{H})$. Observe that  $|\overline{\Obs_{01}^\subset}(\structure{H})|$ is finite too. We know by Lemma~\ref{lem:mpstardichotomy} that, for every  $\structure{G}$ in $\Obs_\star^\subset(\structure{H})$, there exists a $01$-structure $\structure{G}_{01}$ such that:
\begin{itemize}
    \item there is a surjective homomorphism $\pi_\structure{G}\colon\structure{G}_{01} \rightarrow \structure{G}$;
    \item $\structure{G}\not\to\structure{G}_{01}$;
    \item $\structure{G}$ is in $\MP_\star(\structure{H})$ if and only if $\structure{G}_{01}$ is in $\MP(\structure{H})$.
\end{itemize}
As $\structure{G}\notin \MP_\star(\structure{H})$, we can conclude that $\structure{G}_{01}\notin \MP(\structure{H})$. Because $\structure{G}_{01}$ is
a $01$-structure, there exists $\structure{G}_{01}'$ in $\Obs_{01}^\subset(\structure{H})$ such that $\structure{G}_{01}'$ is an induced substructure
of $\structure{G}_{01}$, and thus, by transitivity, $\structure{G}_{01}'\to\structure{G}$. By minimality of $\structure{G}$ (recall that
$\structure{G}$ is in $\Obs_\star^\subset(\structure{H})$), this homomorphism is surjective, i.e., $\structure{G}$ belongs to $\overline{\Obs_{01}^\subset}(\structure{H})$. We
have thus proved that $\Obs_\star^\subset(\structure{H})\subseteq \overline{\Obs_{01}^\subset}(\structure{H})$, i.e., that it is finite.
\end{proof}

By transitivity, we obtain the following.
\begin{corollary}\label{cor:obs_star_fd}
$\Obs_\star^\subset(\structure{H})$ is finite if and only if $\MP_\star(\structure{H})$ has finite duality.
\end{corollary}

\subsection{Looking at $\Cat_\varnothing$}

The goal now is to prove the remaining arrows on the diagram from \cpageref{fig:obs_diagram}. 

\begin{proposition}\label{prop:mp_empty_mp_star}
If $\MP_\varnothing(\structure{H})$ has finite duality, then $\MP_\star(\structure{H})$ has finite duality.
\end{proposition}

\begin{proof}
Let $\family{F}_\varnothing$ be a finite duality set for $\MP_\varnothing(\structure{H})$. Let $\family{F}_\star$ be a duality set for
  $\MP_\star(\structure{H})$. By Proposition~\ref{prop:homimpliesinclusion} and Corollary~\ref{cor:hommin-star-01}, we can assume without loss of generality that $\family{F}_\star=\Obs_{01}^\subset(\structure{H})$. For every $\structure{G}$ in  $\Obs_{01}^\subset(\structure{H})$, we have  $\structure{G}\not\to\structure{H}$, so we have $\structure{G}_\varnothing\to\structure{G}$ for some $\structure{G}_\varnothing$ in $\family{F}_\varnothing$. This homomorphism must be surjective because every proper induced substructure of $\structure{G}$ can be mapped to $\structure{H}$. As $\family{F}_\varnothing$ is finite, there is some constant $c$ such that, for every $\structure{G}_\varnothing$ in $\family{F}_\varnothing$, $|G_\varnothing|<c$. Then, $|G|<c$ for every $\structure{G}$ in $\Obs_{01}^\subset(\structure{H})$ so it is finite.
\end{proof}

We now prove a similar statement for minimal obstructions. 

\begin{proposition}\label{prop:obs_empty_obs_star}
If $\Obs_\varnothing^\subset(\structure{H})$ is finite, then $\Obs_\star^\subset(\structure{H})$ is finite.
\end{proposition}
\begin{proof} Every $\star$-structure is also a $\varnothing$-structure. Thus, $\Obs_\star^\subset(\structure{H})\subseteq \Obs_\varnothing^\subset(\structure{H})$.
\end{proof}

We are now going to prove that $\Obs_\varnothing^\subset(\structure{H})$ is finite if and only if $\MP_\varnothing(\structure{H})$  has finite duality.

\begin{proposition}\label{prop:empty_fd_obs}
If $\MP_\varnothing(\structure{H})$ has finite duality, then $\Obs_\varnothing^\subset(\structure{H})$ is finite.
\end{proposition}
\begin{proof}
Let $\family{F}_\varnothing$ be a finite duality set for $\MP_\varnothing(\structure{H})$ and let $c$ be the maximal size of a structure in $\family{F}_\varnothing$. Consider $\structure{G}$ in $\Obs_\varnothing^\subset(\structure{H})$.  Then, there exists $\structure{T}$ in $\family{F}_\varnothing$ such that $\structure{T}\to\structure{G}$. Moreover, we know that $\structure{T}$ always maps surjectively to $\structure{G}$, because otherwise the substructure of $\structure{G}$ induced by the image of
  $\structure{T}$ would not map to $\structure{H}$, contradicting that $\structure{G}$ is a minimal obstruction. Therefore, $|G|\leq c$, implying that $\Obs_\varnothing^\subset(\structure{H})$ is finite.
\end{proof}

We state the
following, which finishes the proof of the diagram from \cpageref{fig:obs_diagram}.

\begin{proposition}\label{prop:empty_obs_fd}
If $\Obs_\varnothing^\subset(\structure{H})$ is finite, then $\MP_\varnothing(\structure{H})$ has finite duality.
\end{proposition}
\begin{proof}
It is sufficient to show that $\Obs_\varnothing^\subset(\structure{H})$ is a duality set. Suppose that $\structure{G}\not\to\structure{H}$, for some $\structure{G}$ in $\Cat_\varnothing$. If, for all $x\in G$, the substructure induced by $G\setminus \{x\}$ maps to $\structure{H}$, then $\structure{G}$ is already a minimal obstruction. Otherwise, there exists a proper induced substructure that does not map to $\structure{H}$. Then, we can choose one such substructure that is minimal by inclusion. It belongs to $\Obs_\varnothing^\subset(\structure{H})$ and can be mapped to $\structure{G}$ as it is an induced substructure. This means that $\Obs_\varnothing^\subset(\structure{H})$ is a duality set.
\end{proof}

In the following proposition, we show that there exists $\structure{H}$ such that $\Obs_{01}^\subset(\structure{H})$ is finite and
$\Obs_\varnothing^\subset(\structure{H})$ is infinite. Hence, the finiteness of $\Obs_\star^\subset(\structure{H})$ does not
  imply the finiteness of $\Obs_\varnothing^\subset(\structure{H})$ in general. However, one may ask, for which $\star$-graphs $\structure{H}$, the finiteness of $\Obs^\subset_\star(\structure{H})$ implies the finiteness of $\Obs^\subset_\varnothing(\structure{H})$?

\begin{proposition}\label{prop:emptyset_obstruction_infinite}
Let $\structure{H} = \structure{K}_2$ be a $01$-graph, the clique on $2$ vertices. Then, $\Obs_{01}^\subset(\structure{H})$ is finite and $\Obs_\varnothing^\subset(\structure{H})$ is infinite.
\end{proposition}
\begin{proof}
Feder and Hell proved in~\cite{feder_hell_obstructions_full_2008} that once $\structure{H}$ is a $01$-graph, the minimal obstructions for $\MP(\structure{H})$ have bounded size. Thus, $\Obs_{01}^\subset(\structure{H})$ is finite.

Let us show that $\Obs_\varnothing^\subset(\structure{H})$ is infinite. Consider a $\varnothing$-graph $\structure{C}_{n}$ on the domain $v_1,\ldots,v_n$ with $\rel{E}^{\structure{C}_n}(v_i,v_{i+1})=1$ for all $i$ in $[n-1]$ and with $\rel{E}^{\structure{C}_n}(v_n,v_1)=1$, and with all other arcs equal to $\varnothing$. The problem $\structure{C}_n\to\structure{H}$ is equivalent to the $2$-colouring of a directed cycle that is a directed graph, for which we know that odd cycles are all minimal obstructions. Similarly, deleting any vertex from $\structure{C}_n$ creates a $\varnothing$-graph that maps to $\structure{H}$. Thus, the set $\family{C} = \{\structure{C}_n\mid n\text{ is odd}\}$ is an infinite set of minimal obstructions for $\MP_\varnothing(\structure{H})$.
\end{proof}

\section{Remarks on tractability}\label{sec:tractability}

Despite some cases (see for instance~\cite{hell2014,federhellsplit2014,hell_transitive_2017}), the tractability of Matrix Partition Problems on some graph classes is not that studied. Having proved some similarities with usual $\CSP$s, we can ask, for instance, whether well-known graph classes with tractable $\CSP$s still have tractable Matrix Partition Problems. We show that this is unlikely and deserves to be investigated. Let us explain. 

\emph{Tree-width}~\cite{atserias2008,grohe2007} is a well-known graph parameter due to its numerous algorithmic applications, in particular, 
any $\CSP(\structure{H})$ is polynomial time solvable on the class of graphs of bounded tree-width. More importantly, checking whether, for two graphs $\structure{G}$ and $\structure{H}$, there is a homomorphism from $\structure{G}$ to $\structure{H}$, can be solved in time $(|G|+|H|)^{poly(k)}$, where $k$ is the tree-width of
$\structure{G}$.  A natural question is whether such an algorithm exists for Matrix Partitions. One can define the \emph{$\star$-tree-width} of a $\star$-graph as follows.

\begin{definition}
For a $\star$-graph $\structure{G}$, let $\structure{G}_0,\structure{G}_1$ be two directed graphs with the same domain $G$ such that, for $\ast\in\{0,1\}$, $\rel{E}^{\structure{G}_\ast} = \{(x,y)\in G^2\mid \ast \preceq_\star \rel{E}^\structure{G}(x,y)\}$. Then, the \emph{$\star$-tree-width} of $\structure{G}$ is the minimum of the tree-width of $\structure{G}_0$ and the tree-width of $\structure{G}_1$.
\end{definition}

This definition seems natural, because we can describe in $\fo$ the omitted arcs using those that are present. One easily checks by
  Courcelle's theorem (see for instance~\cite{courcelle2012}) that, for every fixed $\star$-graph $\structure{H}$ and every fixed positive integer $k$, we can
  decide in time $f(k)\cdot |G|$, for some function $f:\mathbb{N}\to \mathbb{N}$, whether an input $\star$-graph $\structure{G}$ of $\star$-tree-width $k$
  belongs to $\MP_\star(\structure{H})$. We prove however that, unless $\polynomial=\NP$, there is no algorithm running in time
$(|G|+|H|)^{poly(k)}$, where $k$ is the $\star$-tree-width of $\structure{G}$, and on input $(\structure{G},\structure{H})$ checks whether
$\structure{G}\in \MP_\star(\structure{H})$, even for the case $k=1$.

For a family of $\star$-graphs $\family{G}$, we denote by $\MP_\star(\family{G}, -)$  the set of all pairs of $\star$-graphs $\structure{G},\structure{H}$ such that $\structure{G}\in\family{G}$ and $\structure{G}\to\structure{H}$.

A graph is called a \emph{tree} if its tree-width is equal to 1. Denote by $\family{T}$ the class of all $01$-graphs having $\star$-tree-width equal to 1, called \emph{$01$-trees}. We will prove the following theorem in this section.

\begin{theorem}\label{thm:3sat}
The problem $\MP(\family{T}, -)$ is $\NPc$.
\end{theorem}

The plan is to reduce the 3-SAT problem to $\MP(\family{T},-)$. It will be convenient to represent 3-SAT as a constraint satisfaction problem. For $i,j,k\in\{0,1\}$, put $\rel{R}_{ijk}:=\{0,1\}^3\setminus\{(i,j,k)\}$. Then, the 3-SAT problem can be presented as $\CSP(\{0,1\};\rel{R}_{000},\ldots,\rel{R}_{111})$. The input is given by a primitive-positive sentence $\varphi$ with $N$ variables and $m$ clauses:
\begin{equation*}
\varphi = \exists x_1,\ldots,x_N\bigwedge_{j=1}^m \rel{R}_{s_{j_1}s_{j_2}s_{j_3}}(x_{j_1},x_{j_2},x_{j_3})
\end{equation*}

\subsection{Construction of the 01-tree $\structure{T}$}

Without loss of generality, we can assume that, for all $j\in[m]$, $j_1\leq j_2\leq j_3$. Call $s_{j_1}s_{j_2}s_{j_3}$ the \emph{negation type} of the $j$-th clause. Observe that there are eight possible negation types, assign to them numbers from 1 to 8 as follows: $s_{j_1}s_{j_2}s_{j_3}\mapsto 4{s_{j_1}}+2{s_{j_2}}+s_{j_3}+1$. For example, the negation type of a clause $\rel{R}_{101}(x_1,x_2,x_3)$ is $4\cdot 1+2\cdot 0+1+1=6$.

For $\ell\in[8]$, we introduce a 01-tree $\structure{N}_\ell$. Its domain has 13 elements $\{n_1,\ldots,n_{13}\}$ and the relation $\rel{E}^{\structure{N}_\ell}$ is defined as follows, see \Cref{fig:negs} for some examples.
\begin{align*}
\rel{E}^{\structure{N}_\ell}(n_u,n_v) &= \begin{cases}1 &\text{if [$u+1=v$ and $u\not= \ell+2$] or [$v+1=u$ and $v=\ell+2$],}\\ 0 &\textrm{otherwise.}\end{cases}
\end{align*}

\begin{figure}[ht]
    \centering
    \includegraphics[width=0.35\textwidth]{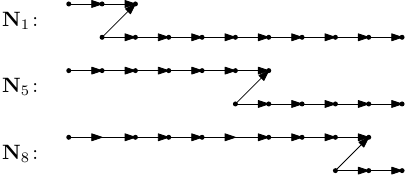}
    \caption{01-trees that represent the negation types.}
    \label{fig:negs}
\end{figure}

For every variable $x_i\in\{x_1,\ldots,x_N\}$, we introduce a 01-tree $\structure{P}_i$. Its domain has $N+5$ elements $\{p_1,\ldots,p_{N+5}\}$ and the relation $\rel{E}^{\structure{P}_i}$ is defined as follows, see \Cref{fig:paths} for some examples.
\begin{align*}
\rel{E}^{\structure{P}_i}(p_u, p_v) &= \begin{cases}1 &\text{if [$u+1=v$ and $u\not= i+2$] or [$v+1=u$ and $v=i+2$],}\\ 0 &\textrm{otherwise.}\end{cases}
\end{align*}

\begin{figure}[ht]
    \centering
    \includegraphics[width = 0.25\textwidth]{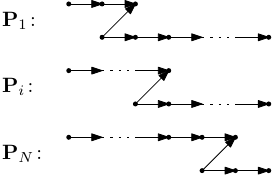}
    \caption{Correspondence between paths and variables.}
    \label{fig:paths}
\end{figure}

Now, we construct the 01-tree $\structure{T}$ rooted at a node $r_T$. The degree of $r_T$ will be equal to $m$ -- the number of clauses in $\varphi$. For each $j\in[m]$, do the following.
\begin{enumerate}
    \item Assuming that the negation type of the $j$-th clause is $\ell\in[8]$, we add to $\structure{T}$ a copy of $\structure{N}_\ell$, denoted by $\structure{N}_j$ and identify its leftmost vertex $n_1^j$ with the root $r_T$.
    \item For the variables $x_{j_1}, x_{j_2}, x_{j_3}$ of the $j$-th clause, we introduce copies of $\structure{P}_{j_1}, \structure{P}_{j_2}, \structure{P}_{j_3}$ and identify their leftmost vertices: $p_1^{j_1}=p_1^{j_2}=p_1^{j_3}$, denoting this vertex by $p_1^j$.
    \item Add a 1-arc between the rightmost vertex $n_{13}^j$ of $\structure{N}_j$ and $p_1^j$: $\rel{E}^\structure{T}(n_{13}^j,p_1^j)=1$. All other vertices between different gadgets are connected by 0-arcs.
\end{enumerate}

Such a  $01$-tree $\structure{T}$ is displayed on \Cref{fig:input}.

\begin{figure}[ht]
    \centering
    \includegraphics[width = 0.35\textwidth]{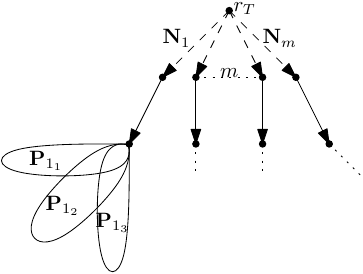}
    \caption{The construction of the  $01$-tree $\structure{T}$.}
    \label{fig:input}
\end{figure}

\subsection{Construction of the 01-graph $\structure{H}$}

$\structure{H}$ is constructed in a similar fashion as $\structure{T}$, we start by introducing a vertex $r_H$ and, for every $j\in[m]$, do the following.

\begin{enumerate}
    \item Assuming that the negation type of the $j$-th clause is $\ell\in[8]$, $\ell=s_{j_1}s_{j_2}s_{j_3}$, we add to $\structure{H}$ a copy of $\structure{N}_\ell$, denoted by $\structure{N}_j$ and identify its leftmost vertex $n_1^j$ with the vertex $r_H$.
    \item For the variables $x_{j_1}, x_{j_2}, x_{j_3}$ of the $j$-th clause, we introduce 7 copies of $\structure{P}_{j_1}, \structure{P}_{j_2}, \structure{P}_{j_3}$, denoted by $\structure{P}_{j_1}^k, \structure{P}_{j_2}^k, \structure{P}_{j_3}^k$, for $k\in[8]\setminus \ell$. For $v\in[3]$, say that the path $\structure{P}_{j_v}^k$ has \emph{label} $\lfloor 2^{v-3}(k-1)\rfloor\mod 2$. For example, if $k=6$, then $\structure{P}_{j_1}^k$ and $\structure{P}_{j_3}^k$ have labels 1, and $\structure{P}_{j_2}^k$ has label 0. This is because 6 is associated with the triple $101$.
    \item For each $k$, we identify the leftmost vertices of $\structure{P}_{j_1}^k, \structure{P}_{j_2}^k, \structure{P}_{j_3}^k$: $p_1^{k,j_1}=p_1^{k,j_2}=p_1^{k,j_3}$, denote this vertex by $p_1^{k,j}$.
    \item For each $k\in[8]\setminus \ell$, add a 1-arc between the rightmost vertex $n_{13}^j$ of $\structure{N}_j$ and $p_1^{k,j}$: $\rel{E}^\structure{H}(n_{13}^j,p_1^{k,j})=1$.
    \item Finally, for every $j,j'\in[m]$, do the following. Assume that the negation type of the $j'$-th clause is $\ell'=s_{j'_1}s_{j'_2}s_{j'_3}$ (remind the negation type of the $j$-th clause is $\ell$). Let $p_{N+5}^{k,j_v},p_{N+5}^{k',j'_{v'}}$ be the rightmost vertices of the paths $\structure{P}_{j_v}^k,\structure{P}_{j'_{v'}}^{k'}$, for all $v,v'\in[3],k\in[8]\setminus \ell$ and $k'\in[8]\setminus \ell'$. For every two such vertices, if $x_{j_v}=x_{j'_{v'}}$ and if the paths $\structure{P}_{j_v}^k$ and $\structure{P}_{j'_{v'}}^{k'}$ have different labels, then add a 1-arc between $p_{N+5}^{k,j_v}$ and $p_{N+5}^{k',j'_{v'}}$: $\rel{E}^\structure{H}(p_{N+5}^{k,j_v},p_{N+5}^{k',j'_{v'}})=1$.
\end{enumerate}
\begin{figure}[ht]
    \centering
    \includegraphics[width = 0.5\textwidth]{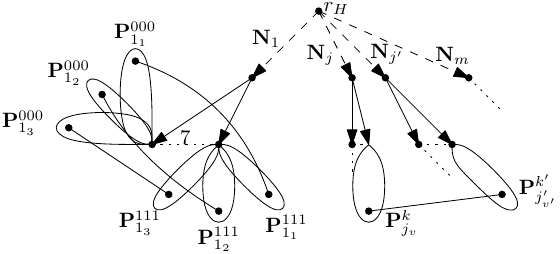}
    \caption{The construction of $\structure{H}$.}
    \label{fig:target}
\end{figure}

As the last step, arcs are added between the right ends of $\structure{P}_{j_v}^k$ and of $\structure{P}_{j'_{v'}}^{k'}$ that are associated with the same variable but have different labels. Their purpose is to forbid the same variable $x_{j_v}=x_{j'_{v'}}$ to take different values in the $j$-th and the $j'$-th clauses. The resulting 01-graph $\structure{H}$ is displayed on \Cref{fig:target}.

\begin{lemma}\label{lem:psize}
The sizes of both $\structure{T}$ and $\structure{H}$ are $O(mN)$, where $m$ is the number of clauses and $N$ is the number of variables in $\varphi$.
\end{lemma}
\begin{proof}
For $i\in[N]$, the size of each $\structure{P}_i$ is $O(N)$. For $\ell\in[8]$, the size of each $\structure{N}_\ell$ is $O(1)$. So the size of $\structure{T}$ is $O(mN)$. $\structure{H}$ is 7 times larger than 
 $\structure{T}$.
\end{proof}

Before proving the next lemma, we need to define the \emph{relative height} function for $\structure{T}$. It is defined inductively. First, $height_T(x,x):=0$, for all $x\in T$. Suppose that $height_T(x,y)$ is known and equal to $n$. If, for some $z\in T$, $\rel{E}^\structure{T}(y,z)=1$, then $height_T(x,z):=n+1$. If $\rel{E}^\structure{T}(z,y)=1$, then $height_T(x,z)=n-1$. Observe that $height_T$ is well-defined because $\structure{T}$ is a 01-tree and that it has its maximum at pairs of the form $(r_T, p_{N+5}^{j_v})$, where $p_{N+5}^{j_v}$ is the rightmost endpoint of  $\structure{P}_{j_v}$, for some $j\in[m]$ and $v\in[3]$.

\begin{lemma}\label{lem:root2root}
For every homomorphism $h\colon \structure{T}\to\structure{H}$, we have that $h(r_T)=r_H$.
\end{lemma}
\begin{proof}
First, observe that, for every arc $(p_{N+5}^{k,j_v},p_{N+5}^{k',j'_{v'}})$ that we added between the ends of some $\structure{P}$ and
  $\structure{P}'$ at the final step of the construction of $\structure{H}$, the preimage of this arc is empty, i.e., for at least one of
  $p_{N+5}^{k,j_v}$ and $p_{N+5}^{k',j'_{v'}}$, there is no vertex in $\structure{T}$ that is mapped to it by $h$. Indeed, for no two vertices $x,x'$ of
  $\structure{T}$, we have that $\rel{E}^\structure{T}(x,x')=\rel{E}^\structure{T}(x',x)=1$. Then, we can assume that the image $h(\structure{T})$ is an induced
  substructure of $\structure{H}$.

For every two elements $x,x'$ of $h(\structure{T})$, it is impossible that $\rel{E}^\structure{H}(x,x')=\rel{E}^\structure{H}(x,x')=1$. Therefore,
  the function $height_T$ is well-defined on $h(\structure{T})$, and also it is preserved by $h$. Denote this function by $height_H$. By the construction of
  $\structure{H}$, the only possible case, when $height_H(x,x') = height_T(r_T, p_{N+5}^{j_v})$ is when $x = r_H$ and $x' = p_{N+5}^{k,j'_{v'}}$, where
  $p_{N+5}^{j_v}$ is the rightmost endpoint of some path $\structure{P}_{j_v}$ of $\structure{T}$ and $p_{N+5}^{k,j'_{v'}}$ is the rightmost endpoint of a
  similar path in $\structure{H}$.
\end{proof}

\begin{proof}[Proof of Theorem~\ref{thm:3sat}]
 Let $f\colon\{x_1,\ldots,x_N\}\to\{0,1\}$ be a valid assignment for the variables of $\varphi$. By Lemma~\ref{lem:psize}, construct in $\polynomial$-time a $01$-tree $\structure{T}$ and a $\star$-graph $\structure{H}$ as above. Take the $j$-th clause and map the corresponding part of $\structure{T}$ to the triple of paths labelled with $f(x_{j_1})f(x_{j_2})f(x_{j_3})$. By  construction, it will be a homomorphism.
 
 Let $h\colon \structure{T}\to \structure{H}$ be a homomorphism. We know, by Lemma~\ref{lem:root2root}, that $h(r_T)=r_H$. By
 construction, for all $i,i'\in[N], \structure{P}_i\to\structure{P}_{i'}$ if and only if $i=i'$. Similar statement holds for
   $\structure{N}_1,\ldots,\structure{N}_8$. Therefore, we know that the part of $\structure{T}$ corresponding to the $j$-th clause is mapped to
 the part of $\structure{H}$ corresponding to the same clause of $\varphi$. Let us now construct a valid assignment $f$ of $\varphi$. We
   first observe, in the construction of $\structure{H}$, that the seven paths associated with a variable of a clause correspond to the $7$ valid assignments of
   that clause. So, define the label of such a path as the value of the corresponding variable in that valid assignment of the clause. Now, observe that for
   each variable $x_i$, and each $j,j'\in [m]$ such that $x_i$ belongs to the $j$-th and $j'$-th clauses, the path in $\structure{T}$ associated with the
   variable $x_i$ within the $j$-th clause is mapped to a path labelled $0$ in $\structure{H}$ if and only if the path in $\structure{T}$ associated with the variable $x_i$
   within the $j'$-th clause is mapped to a path labelled $0$ in $\structure{H}$. Therefore, all the paths corresponding to a variable $x_i$ in $\structure{T}$
   are mapped to paths of the same label: either all to 0 or all to 1. So, we can define $f(x_i)$ to be equal to this unique label. Now, this assignment $f$
   will be valid because as we observed above, whenever a path in $\structure{T}$, corresponding to a variable $x_i$ in a clause $j$, is mapped to a path in
   $\structure{H}$, this path in $\structure{H}$ corresponds to a valid assignment for the clause $j$.
\end{proof}

\begin{remark}
Observe that the $01$-tree $\structure{T}$ that represents a 3-SAT formula also has bounded pathwidth, so the result of Theorem~\ref{thm:3sat} will remain true if we require that all the $01$-trees of $\family{T}$ have bounded pathwidth.
\end{remark}

\section{Conclusion}

We have proposed several generalisations of the Matrix Partition Problems studied by Hell et al.  We have shown that $\MP$ and $\MP_\star$ are
$\polynomial$-time equivalent and we have used this to show that a dichotomy for every class $\MP^{\tilde{\sigma}}$ with $|\tilde{\sigma}|=1$ implies a
dichotomy for $\MP^\sigma$ for every finite $\sigma$.  Despite this, we leave open the question of whether $\MP$ on directed graphs is $\polynomial$-time
equivalent to $\MP^\sigma$, for every finite signature $\sigma$, and, a fortiori, the dichotomy question for $\MP$.  We have introduced the generalisation
$\MP_\varnothing$ as a way to see $\MP$ as a $\CSP$ on ``complete inputs''.  We have also studied the set of minimal
obstructions proposed by Feder et al.~\cite{federhell2007} and have proved that their finiteness coincides with finite duality for $\MP$ and $\MP_\star$
problems.  This, we believe, would allow to characterise the finiteness of minimal obstructions for $\MP$ problems.  Finally, we have shown the difference
between $\MP$ and $\CSP$ with respect to the bounded tree-width input by reducing 3-SAT to $\MP(\family{T},-)$.

\newpage

\appendix

\section{Proof of Proposition~\ref{prop:twotoincomplete}}
\label{apsection:from_3_to_2}
Let $\structure{G}_3$ be an $(\varnothing,\tilde{\sigma})$-structure and $\structure{G}_2$ be a $\star$-graph constructed from it by replacing tuples with values in $\{0,1,\star\}$ with copies of gadgets $\structure{T}^0,\structure{T}^1,\structure{T}^\star$ such that they satisfy the conditions 1--5 from Section \ref{subsec:onetobinary}; and $\varnothing$-valued tuples are just deleted without being replaced by anything.
If there is a homomorphism $h_3\colon\structure{G}_3\to\structure{H}_3$, then, by the condition 3, there is a homomorphism $h_2\colon\structure{G}_2\to\structure{H}_2$ such that $h_3$ is the restriction of $h_2$ onto $H_3$. If there is a homomorphism $h_2\colon\structure{G}_2\to\structure{H}_2$, then, by the condition 1, one can consider the restriction $h_3$ of this map on the set $G_3$, and the codomain of this map will be the set $H_3$. By the condition 3, $h_3$ is a homomorphism from $\structure{G}_3$ to $\structure{H}_3$.

Now, consider any $\star$-graph $\structure{A}$ from the input of $\MP_\star(\structure{H}_2)$. By the condition 2, we can mark, in $\polynomial$-time, all the
elements of $\structure{A}$ that can be mapped only to the elements of $H_3$, denote the set containing them by $A_3$. Let $eq(\cdot,\cdot)$ be the minimal by inclusion equivalence relation on the set $A\setminus A_3$ such that, for $a_0,a_n\in A\setminus A_3$, $eq(a_0,a_n)$ holds if there exists a sequence of elements
$a_0,a_1,\ldots,a_{n-1},a_n\in A\setminus A_3$ such that for every $0\leq i<n$, either $\rel{E}^\structure{A}(a_i,a_{i+1})\not=0$ or
$\rel{E}^\structure{A}(a_{i+1},a_i)\not=0$. For each $eq$-equivalence class $A_a$ (containing an element $a$), let $\structure{A}_a$ be the subgraph of $\structure{A}$ induced by the set
\begin{equation*}
A_a\cup\{b\in A_3 \mid \exists c\in A_a\; \rel{E}^\structure{A}(b,c)\not=0 \text{ or } \rel{E}^\structure{A}(b,c)\not=0\}
\end{equation*}
Call each such $\structure{A}_a$ an \emph{internal component} of $\structure{A}$.
Below we will show that the image of every internal component $\structure{A}_a$ must be
contained in some $\structure{T}^v_{xyz}$.

\begin{claim}
If there is $h\colon \structure{A}_a\to\structure{H}_2$, then $h(A_a)\subseteq T^v_{xyz}$ for some $\structure{T}^v_{xyz}$.
\end{claim}
\begin{proof}[Proof of the claim]
  For every two elements $a_0,a_n$ of $A_a$, there exists a sequence $a_1,\ldots,a_{n-1}$ of elements of $A_a$ such that, for each
  $i\in[n-1]$, one of $\rel{E}^{\structure{A}}(a_i,a_{i+1})$ and $\rel{E}^{\structure{A}}(a_{i+1},a_i)$ is not $0$. 
  By condition 2, all the elements of this sequence are mapped to $H_2\setminus H_3$. Since, for all $i\in[n-1]$, one of $\rel{E}^{\structure{H}_2}(h(a_i),h(a_{i+1}))$ and $\rel{E}^{\structure{H}_2}(h(a_{i+1}),h(a_i))$ is not equal to $0$, we see that $h(a_1),\ldots,h(a_{n-1})$ must belong to the same $T^v_{xyz}$, by condition 4. This implies that $a_0$ and $a_n$ belong to the same $T^v_{xyz}$.
\end{proof}

By the condition 3, for every internal component $\structure{A}_a$, we find in $\polynomial$-time the values $v\in\{0,1,\star\}$ such that $\structure{A}_a$ maps to
$\structure{T}^v$. We are going to assign to each component $\structure{A}_a$ a value from the set $\{\varnothing,0,1,\star,\not\to\}$ by the following rules.
\begin{itemize}
    \item If $\structure{A}_a$ maps to $\structure{T}^v$ for any possible $v$, then we assign ``$\varnothing$'' to it.
    \item If, for all $v\in\{0,1,\star\}$, $\structure{A}_a\not\to\structure{T}^v$, then we assign ``$\not\to$'' to it.
    \item Otherwise, among all $v\in\{0,1,\star\}$ such that $\structure{A}_a\to\structure{T}^v$, we label $\structure{A}_a$ with the smallest possible such ``$v$'' with respect to $\preceq_\star$.
\end{itemize}
If there is at least one ``$\not\to$''-labelled component $\structure{A}_a$, then there is no homomorphism from $\structure{A}$ to $\structure{H}_2$, so further we assume that there are no ``$\not\to$''-labelled subgraphs.
Let $\sim$ be the minimal by inclusion equivalence relation on the set $A_3$ defined as follows. For $a_1,a_2\in A_3$, $a_1\sim a_2$ if there
exist an internal component $\structure{A}_a$ labelled with $v\in\{0,1,\star\}$ such that
\begin{itemize}
    \item both $a_1,a_2$ belong to $\structure{A}_a$, and
    \item there is a homomorphism $h\colon \structure{A}_a\to\structure{T}^v$ such that $h(a_1)=h(a_2)$.
\end{itemize}
By the condition 5, for every $a_1,a_2\in A_3$ and every homomorphism $h\colon \structure{A}\to\structure{H}_2$, $a_1\sim a_2$ implies $h(a_1)=h(a_2)$.

Let us construct a new $\star$-graph
$\structure{B}$ based on $\structure{A}$. First, take the set $B = A_3/\sim$.

For every internal component $\structure{A}_a$ with a label $v\in\{0,1,\star\}$ and for every $w\in\{x,y,z\}$, denote by $X_{a,w}$ the set of all elements of $\structure{A}_a$ which are also in $A_3$ that are mapped to $w$ by every homomorphism $h\colon \structure{A}_a\to\structure{T}^v_{xyz}$. The set $X_{a,w}$ is well-defined by condition 5. However, it may happen that $X_{a,w}$ can be empty. In each such case, we add to $B$ and to $X_{a,w}$ a newly introduced element $b_{a,w}$. The relation $\sim$ is extended on these new elements as equality, i.e., every $b_{a,w}$ is equivalent only to itself.

Say that an internal component $\structure{A}_a$ with label $v\in\{0,1,\star\}$ is \emph{spanned} by $(b_1,b_2,b_3)\in B^3$ if $X_{a,x},X_{a,y}, X_{a,z}$ belong to equivalence classes of $\sim$ associated with $b_1,b_2,b_3$ respectively.

For each $(b_1,b_2,b_3)\in B^3$, define a set $V_{b_1b_2b_3}\subseteq\{0,1,\star\}$ to be the set of all $v$'s, where $v$ is the label of an internal component spanned by $(b_1,b_2,b_3)$.
For each $(b_1,b_2,b_3)$, do the following.
\begin{itemize}
    \item If $V_{b_1b_2b_3}$ is empty, then do nothing.
    \item Otherwise, add a copy of a gadget $\structure{T}^{v_{b_1b_2b_3}}_{b_1b_2b_3}$ such that it contains $b_1,b_2,b_3$ and that $v_{b_1b_2b_3} = \bigvee_{v\in V_{b_1b_2b_3}}v$.
\end{itemize}
It follows from the construction of $\structure{B}$ that $\structure{A}\to\structure{H}_2$ if and only if $\structure{B}\to\structure{H}_2$. Also, one may easily check that $\structure{B}$ can be associated with some $(\varnothing,\tilde{\sigma})$-structure $\structure{B}_3$ such that $\structure{B}\to\structure{H}_2$ if and only if $\structure{B}_3\to\structure{H}_3$. So, we can reduce $\structure{A}$ to such a structure $\structure{B}_3$, and we are done.

\section*{Acknowledgement}
The authors thank Florent Madelaine for the many fruitful discussions on the subject, and the reviewers for all their comments which helped improve the
manuscript. 


\bibliographystyle{alpha} 
\bibliography{revision2.bib}






\end{document}